\documentclass[11pt,a4paper,reqno]{article}

\usepackage{amssymb}
\usepackage{latexsym}
\usepackage{amsmath}
\usepackage{graphicx}
\usepackage{amsthm}
\usepackage{empheq}
\usepackage{bm}
\usepackage{booktabs}
\usepackage[dvipsnames]{xcolor}
\usepackage{pagecolor}
\usepackage{subcaption}
\usepackage{tikz,lipsum,lmodern}
\usepackage{enumitem}
\usepackage{calligra}
\usepackage{mathrsfs}
\usepackage[margin=1in]{geometry}
\usepackage{authblk}
\usepackage{enumitem}
\setitemize{itemsep=0em}
\setenumerate{itemsep=0em}
\usepackage{hyperref}
\hypersetup{
  colorlinks   = true, 
  urlcolor     = blue, 
  linkcolor    = purple, 
  citecolor   = ForestGreen 
}
\usepackage{orcidlink}

\numberwithin{equation}{section}

\theoremstyle{definition}
\newtheorem{theorem}{Theorem}[section]
\newtheorem{corollary}[theorem]{Corollary}
\newtheorem{proposition}[theorem]{Proposition}
\newtheorem{definition}[theorem]{Definition}
\newtheorem{example}[theorem]{Example}

\newtheorem{remark}[theorem]{Remark}
\newtheorem{lemma}[theorem]{Lemma}

\newtheorem{theoremx}{Theorem}

\renewcommand{\emph}[1]{{\textbf{#1}}}

\allowdisplaybreaks

\def\F{\mathbb{F}}

\def\rad{\textup{rad}}
\def\irk{\textup{irk}}
\def\spn{\textup{span}}
\def\wt{\textup{wt}}
\def\maxwt{\textup{maxwt}}
\def\mA{\mathcal{A}}

\setlist[enumerate]{label=(\arabic*), itemsep=5pt, topsep=5pt}
\usepackage{thmtools, thm-restate}

\title{\textbf{Quantum Anticodes}}

\author[1,2]{ChunJun Cao\protect\footnotemark[2]}
\affil[1]{Department of Physics, Virginia Tech, Blacksburg, VA, USA}
\affil[2]{Virginia Tech Center for Quantum Information Science and Engineering, Blacksburg, VA, USA}

\author[2,3]{Giuseppe Cotardo\protect\footnotemark[3]}
\affil[3]{Department of Mathematics, Virginia Tech, Blacksburg, VA, USA}

\author[4,5]{Brad Lackey}
\affil[4]{Microsoft Quantum, Redmond, WA, USA}
\affil[5]{Department of Physics, University of Washington, Seattle, WA, USA}

\date{}

\usepackage{blkarray}

\begin{document}

\maketitle\vspace{-4ex}
\footnotetext[2]{C. C. was partially supported by NSF MPS-2542568 and DOE Genesis Mission DE-SC0026692.}
\footnotetext[3]{G. C. was partially supported by NSF MPS-2554002 and DOE Genesis Mission DE-SC0026692.}
	
\abstract{This work introduces a symplectic framework for quantum error correcting codes in which local structure is analyzed through an anticode perspective. In this setting, a code is treated as a symplectic space, and anticodes arise as maximal symplectic subspaces whose elements vanish on a prescribed set of components, providing a natural quantum analogue of their classical counterparts. This framework encompasses several families of quantum codes, including stabilizer and subsystem codes, provides a natural extension of generalized distances in quantum codes, and yields new invariants that capture local algebraic and combinatorial features. The notion of anticodes also naturally leads to operations such as puncturing and shortening for symplectic codes, which in turn provide algebraic interpretations of key phenomena in quantum error correction, such as the cleaning lemma and complementary recovery. It also yields new descriptions of weight enumerators, including a quantum analogue of Greene's theorem relating weight enumerators to polymatroid invariants.
}

\section*{Introduction}

Quantum computing promises computational advantages over classical methods, yet the realization of large scale architectures fundamentally depends on protecting quantum information from noise. Quantum error correcting codes play a central role in this setting and provide the fundamental framework through which quantum information can be stored and manipulated with high fidelity. While most experimentally implemented codes, such as the surface code~\cite{kitaev2003fault,fowler2012surface}, encode a single logical qubit, scalable quantum systems require codes capable of encoding multiple logical qubits. In such settings, the different logical qubits may exhibit distinct distances, as observed for holographic codes~\cite{pastawski2015holographic}. Standard quantum coding theory, however, is not designed to capture these finer structural features, since it condenses all logical qubits into a single minimum distance parameter. Insights from the Anti-de Sitter/Conformal Field Theory (AdS/CFT) correspondence \cite{maldacena1999large,witten1998anti} further highlight this limitation, showing that finite rate holographic codes can display constant minimum distance while achieving nontrivial distance scaling for logical qubits situated deeper in the bulk~\cite{almheiri2015bulk, pastawski2015holographic, cao2021approximate, steinberg2025universal}.

These observations point toward the need for a framework that can describe how different parts of a code behave under local constraints. For example, a generalized notion of distance is needed to describe how distinct sets of logical qubits can have different distance scalings and operator supports. Stabilizer codes provide a natural context in which such questions arise. They exhibit structural properties that align with this refined viewpoint, most notably complementary recovery \cite{harlow2017ryu, pollack2022understanding}, where complementary subsystems can reconstruct commuting code subalgebras whose union is as large as possible.  This property plays a central role in erasure correction because partial erasure still permits partial recovery of the encoded quantum information. Stabilizer codes are well established as the quantum analogue of classical linear codes~\cite{calderbank1998quantum}, and their connection to symplectic geometry was already established~\cite{gottesman1996class, gottesman1997stabilizer}. Although stabilizer codes were originally framed in the language of projective representations, they are now typically expressed using classical codes of twice the length with specific self duality constraints~\cite[Chapter~10]{nielsen2010quantum}. However, fully understanding their structural properties requires a more fine-grained measure of how distances vary across logical qubits and how logical operators are distributed across subsystems. Finally, when a code is realized through a tensor network \cite{ferris2014tensor, farrelly2021tensor, cao2022quantum}, the underlying geometry imposes locality constraints, making local code invariants essential for accurately predicting global code properties \cite{cao2023quantum,cao2024quantum}. This highlights the limitations of standard metrics and the need for a refined geometric framework that capture local properties of a code.

We introduce quantum anticodes, which are characterized as free codes supported on prescribed subsets of coordinates, providing a natural analogue of classical anticodes within the symplectic setting. A distinguishing feature of this setting is the simultaneous presence of \textit{symplectic dimension} and \textit{isotropic rank} (or simply \textit{isorank}), which capture different aspects of the structure of a code. Their interplay with the radical and symplectic orthogonality provides the algebraic foundation for the duality relations and generalized invariants developed below.

Our approach builds on the theory of generalized weights of linear codes. Generalized Hamming weights were introduced by Wei~\cite{wei1991generalized}, who showed that the generalized weights of a code and those of its dual determine each other. In~\cite{ravagnani2016generalized}, it was shown that generalized weights in several metrics can be defined through intersections with optimal anticodes, and this approach was further developed for rank-metric and tensor codes in~\cite{ravagnani2016rank,byrne2020rank,byrne2023tensor}; see also~\cite{xu2022galois} for an order-theoretic treatment via Galois connections, and~\cite{digiusto2025distance} for a framework based on intersections with general families of test spaces. Our setting differs from these in two respects. First, the alphabet $V$ is a symplectic space, and the relevant notion of size is the symplectic dimension rather than the $\F$-linear dimension. In particular, the dimension is not modular on subspaces, and our invariants $\alpha_C$ and $\beta_C$ introduced below depend on the symplectic structure of the intersections, not only on their $\F$-dimensions. Second, the codes of interest are stabilizer codes, whose radical $C^\perp$ is isotropic. As a consequence, the classical duality between a code and its dual takes the form of a relation between two invariants of the same code (Lemma~\ref{lemma:weight-complementarity}). On the enumerator side, our binomial moments are related to the quantum weight enumerators of~\cite{shor1997quantum,rains2002quantum}.

Moreover, the connection between weight enumerators of classical codes and matroid invariants goes back to Greene~\cite{greene1976weight}, who showed that the weight enumerator of a linear code is determined by the Tutte polynomial of the associated matroid. This was later extended to codes over modules, via latroids, in~\cite{vertigan2004latroids}, and to other supports in~\cite{jurrius2013codes,gorla2025latroids}. Theorem~E below provides an analogue of this result for symplectic codes.

We now summarize our main results. Let $V$ be a $2$-dimensional symplectic space over a field $\F$, spanned by a symplectic pair $(e,f)$, and endow $V^n$ with the induced symplectic form. A \emph{symplectic code} is an $\F$-linear subspace $C\leq V^n$. Its \emph{dimension} $\dim(C)$ is the maximum number of pairwise orthogonal symplectic pairs in $C$, and its \emph{isorank} $\irk(C)$ is the maximum $\F$-dimension of an isotropic subspace of $C$. Stabilizer codes are precisely the codes with $\irk(C)=n$; in this case $\rad(C)=C^\perp$ is the stabilizer and $\dim(C)$ is the number of logical qubits. Following~\cite{ravagnani2016generalized}, we say that $C$ is an \emph{anticode} if $\dim(C)$ equals the maximum Hamming weight $\maxwt(C)$ of its elements. Our first result characterizes these mathematical objects.

\begin{theoremx}[Corollary~\ref{corollary:anticode-bound} and Proposition~\ref{prop:anticode-char}]
    For every code $C\leq V^n$, we have $\dim(C)\leq\maxwt(C)$. Moreover, equality holds if and only if
    \begin{equation*}
        C=A_J=\{v\in V^n : v_j=0 \textup{ for all } j\notin J\}
    \end{equation*}
    for some $J\subseteq\{1,\ldots,n\}$.
\end{theoremx}

Anticodes thus correspond to subsets of coordinates, and intersecting a code with them captures its local structure. For $A=A_J$, the \emph{puncturing} $\Pi_A C$ and the \emph{shortening} $\Sigma_A C$ are the projections onto the coordinates indexed by $J$ of $C$ and of $C\cap A$, respectively.

\begin{theoremx}[Theorem~\ref{thm:cleaning-lemma}]
    For every code $C\leq V^n$ and every $J\subseteq\{1,\ldots,n\}$, with $A=A_J$, we have
    \begin{equation*}
        \Sigma_A C^\perp=(\Pi_A C)^\perp\qquad\textup{ and }\qquad\Pi_A C^\perp=(\Sigma_A C)^\perp,
    \end{equation*}
    where the orthogonals are taken in $V^{|J|}$.
\end{theoremx}

Theorem~B is the symplectic analogue of the classical duality between puncturing and shortening~\cite{huffman2003fundamentals}, and it contains the cleaning lemma of~\cite{bravyi2009no} as a special case. Combined with a complementarity property of the radical, it yields identities relating the complementary subsystems $J$ and $J^c$ of a stabilizer code (Corollary~\ref{cor:stab-complementarity}). For stabilizer states, these recover the equality of the entanglement entropies of complementary subsystems~\cite{fattal2004entanglement}.

Now let $C$ be a stabilizer code with $\dim(C)=k$ and minimum distance $d$. For every anticode~$A$, we consider $\alpha_C(A)=\dim(C\cap A)$, which counts the logical qubits whose logical operators can be supported on $A$, and $\beta_C(A)=\irk(C\cap A)-\irk(\rad(C)\cap A)$. These invariants give rise to generalized weights $\vartheta_a(C),\varphi_a(C)$, generalized profiles $\theta_b(C),\phi_b(C)$, in a natural way in Section~\ref{sec:invariants}. In particular, $\varphi_1(C)=d$, and we obtain the following generalization of the quantum Singleton bound.

\begin{theoremx}[Theorems~\ref{thm:gensingletonboun} and~\ref{thm:qsingleton}]
    For every $a\in\{1,\ldots,k\}$, we have
    \begin{equation*}
        \varphi_a(C) \leq n - d - k + a + 1.
    \end{equation*}
    In particular, $2(d-1)\leq n-k$.
\end{theoremx}

Finally, when $\F=\F_q$, we define the binomial moments $B_b(C)=\sum_{|J|=b}|C\cap A_J|$ and show that they determine the weight enumerator $\mathcal{B}(x,y;C)=\sum_{c\in C}x^{\wt(c)}y^{n-\wt(c)}$.

\begin{theoremx}[Proposition~\ref{prop:binomial-to-enumerator} and Theorem~\ref{thm:macwilliams-moments}]
    Let $C\leq V^n$ be a code with $\dim(C)=k$ and $\irk(C)=s$. Then
    \begin{equation*}
        \mathcal{B}(x,y;C)=\sum_{b=0}^n B_b(C)\,x^b(y-x)^{n-b}\qquad\textup{ and }\qquad B_b(C^\perp)=q^{2b-k-s}B_{n-b}(C)
    \end{equation*}
    for every $b\in\{0,\ldots,n\}$.
\end{theoremx}

In particular, the weight enumerators of a stabilizer code in the sense of~\cite{shor1997quantum} are determined by the intersections of $C$ and $C^\perp$ with anticodes. In fact, these enumerators are determined by a single function on anticodes. For any anticode $A$, let $\rho_C(A)=\dim_{\F_q}(\Pi_A C)$ be the vector space dimension of the puncturing of $C$ on $A$, and consider the polynomial
\begin{equation*}
    R_C(u,v)=\sum_{A\in\mA(n)}u^{\dim_{\F_q}(C)-\rho_C(A)}\,v^{2\dim(A)-\rho_C(A)}.
\end{equation*}
The function $\rho_C$ is the rank function of a polymatroid, and $R_C(u,v)$ plays the role of its Whitney rank generating function. The following result is a quantum analogue of Greene's theorem~\cite{greene1976weight}, which expresses the weight enumerator of a classical linear code in terms of the Tutte polynomial of its matroid.

\begin{theoremx}[Proposition~\ref{prop:poly-dual} and Theorem~\ref{thm:greene}]
    Let $C\leq V^n$ be a code with $\rho=\dim_{\F_q}(C)=\rho_C(V^n)$. Then
    \begin{equation*}
        \mathcal{B}(x,y;C)=x^{n-\rho/2}(y-x)^{\rho/2}\,R_C\!\left(\sqrt{\frac{q^2x}{y-x}},\ \sqrt{\frac{y-x}{x}}\right)\qquad\textup{ and }\qquad R_{C^\perp}(u,v)=R_C(v,u).
    \end{equation*}
\end{theoremx}

As in the classical case, combining the two identities of Theorem~E yields the MacWilliams identity for weight enumerators of symplectic codes (Corollary~\ref{cor:macwilliams-enum}).

The paper is organized as follows. Section~1 collects the necessary background on symplectic spaces, together with the properties of dimension and isorank used throughout. In Section~2, we introduce symplectic codes, explain how stabilizer, subsystem, and entanglement-assisted codes fit into this framework, and prove Theorems~A and~B, together with the complementarity relations for stabilizer codes. Section~3 is devoted to the invariants~$\alpha_C$ and $\beta_C$, the associated generalized weights and profiles, and the proof of Theorem~C. Finally, in Section~4 we study binomial moments and weight distributions, and prove Theorems~D and~E.

Beyond the results above, the framework developed here provides the algebraic and combinatorial foundation for further applications of the anticode approach to quantum error-correcting codes, which will be pursued in forthcoming work.

\section{Symplectic Spaces}

In this section, we recall some definitions and properties of symplectic spaces, following the general reference~\cite{artin2016geometric}, and we derive several preliminary results that will be used in the subsequent sections. Throughout this section, we let $n$ be a non-negative integer and we let~$V$ be a $(2n)$-dimensional vector space over a field $\F$. We use the term \textit{subspace} or the symbol~$\leq$ to mean $\F$-linear subspace.

\begin{definition}
    Let $\omega : V \times V \to \F$ be a bilinear map. We say $\omega$ is \emph{nondegenerate} if~$\omega(u, v) = 0$ for all $v \in V$ implies that $u = 0$. We say $\omega$ is \emph{skew-symmetric} if $\omega(u, u) = 0$ for all $u \in V$ and we refer to the pair $(V, \omega)$ as a \emph{symplectic space}.
\end{definition}

Although this is not the usual definition of skew-symmetry, when  $\operatorname{char}(\F)\neq 2$ it is equivalent to the identity  $\omega(u,v) = -\omega(v,u)$ for all $u,v \in V$; this follows immediately by expanding $0 = \omega(u+v,\,u+v)$. However in characteristic $2$ the usual definition of skew-symmetry reduces to symmetry, which does not in general imply $\omega(u, u) = 0$, which motivates the definition we have adopted. Throughout the paper, we let $(V, \omega)$ be a symplectic space and, by abuse of notation, write~$V$ to denote the symplectic space $(V, \omega)$. 

\begin{definition}
  Let $e,f\in V$. We say that $(e,f)$ is a \emph{symplectic pair} if $\omega(e,f) = 1$.
\end{definition}

For any $W\leq V$, we denote by $\left.\omega\right|_W$ the \emph{restriction} of~$\omega$ on $W$. We recall that, while~$\omega$ is nondegenerate on $V$, the same is generally not true for $\left.\omega\right|_W$.

\begin{definition}
    Let $W$ be a subspace of $V$. We say that $W$ is \emph{symplectic} if $\left.\omega\right|_W$ is nondegenerate. We say that $W$ is \emph{isotropic} if $\left.\omega\right|_W=0$.
\end{definition}

In other words, $W$ is isotropic if every vector in $W$ is orthogonal to every other vector in $W$. Note that this definition is not exhaustive; a general subspace $W$ will be neither symplectic nor isotropic, but will have a symplectic part and an isotropic part, as we will shortly see. 

\begin{remark}
Observe that if $W$ is isotropic then each of its subspaces is also isotropic. The same does not hold for the symplectic property. However, the symplectic property is \emph{stable} in that if $W \leq V$ is symplectic as a subspace of $V$, then it remains symplectic when viewed as a subspace of $U$, for any $W\leq U \leq V$.
\end{remark}

For a subspace $W\leq V$, we denote by $W^\perp$ the orthogonal of $W$ in $V$ with respect to~$\omega$. We introduce the following definitions and results, following~\cite{artin2016geometric}.

\begin{definition}
    The \emph{radical} of $W\leq V$ is $\rad(W)=W\cap W^\perp$.
\end{definition}

For $v\in W$, we say that $v$ is \textbf{isotropic in} $W$ if $v\in\rad(W)$, that is, if $v$ is orthogonal to every element of $W$. It is clear that the radical is an isotropic subspace of $W$. It is worth noting that, in the context of coding theory, the radical is often referred to as the \textit{hull}. This concept was first introduced in connection with coding theory and combinatorial designs, where it played a key role in determining when a code gives rise to a design and in classifying its structural properties (see~\cite{assmus1990affine,assmus1994designs}). More recently, in~\cite{anderson2024relative}, the hull was employed to investigate the properties and establish the existence of new optimal entanglement-assisted quantum error-correcting codes.

\begin{theorem}{{\cite[Theorems~1.3 and~3.3]{artin2016geometric}}}
    Let $W$ be a subspace of $V$. We have that~$W$ admits an orthogonal decomposition (or splitting) $W = \rad(W) \oplus K$, for some symplectic subspace~$K \leq W$.
\end{theorem}

Since $\rad(W) \oplus K$ is an \textbf{orthogonal splitting} of $W$,  the invariants introduced below are well-defined. In the remainder we denote by $\dim_\F(W)$ the dimension of $W$ as a $\F$-linear space and by $K$ a symplectic subspace of $W$.

\begin{definition}\label{def:dim-irk}
    Let $W\leq V$ and let $W=\rad(W)\oplus K$ be an orthogonal splitting of~$W$. We define the \emph{dimension} and the \emph{isotropic rank} (or \emph{isorank}) of $W$ to be respectively 
    \begin{equation*}
        \dim(W)=\tfrac{1}{2}\dim_\F(K),\ \textup{ and }\ \irk(W)=\tfrac{1}{2}\dim_\F(K) + \dim_\F(\rad(W)).
    \end{equation*}
\end{definition}

\begin{remark}\label{rem:dimirk}
Since $\rad(W)$ is orthogonal to $W$, every symplectic subspace of $W$ intersects $\rad(W)$ trivially. Similarly, every isotropic subspace $I\le W$ is contained in the isotropic subspace $I+\rad(W)$, whose projection onto $K$ along $\rad(W)$ is isotropic in $K$. It follows that $\dim(W)$ is the maximum number of pairwise orthogonal symplectic pairs in $W$, and~$\irk(W)$ is the maximum $\F$-dimension of an isotropic subspace of $W$. In particular, neither quantity depends on the chosen splitting.   
\end{remark}

We acknowledge that defining the dimension of $W$ as something other than its vector space dimension may cause some confusion. Nonetheless, this definition of dimension coincides with the dimension of $W$ when viewed as a quantum~code: the number of encoded logical qubits. Throughout the paper, we will consistently use $\dim_\F(\cdot)$ to denote the dimension as a~$\F$-linear space, and reserve $\dim(\cdot)$ for the dimension as defined above. In particular, $V$ satisfies $\dim_\F(V)= 2n$ and $\dim(V) = \irk(V)=n$. We introduce the following two examples, which will be used throughout the paper to illustrate and clarify some of the main concepts.

\begin{example}\label{example:repetition_code1}
    Let $W = \mathrm{span}_{\mathbb{F}}\{u,v,w\} \leq V=\mathbb{F}^4$, with $\omega(u,v) = 1$ and  $\omega(u,w) = \omega(v,w) = 0$. We have $\rad(W) = \mathrm{span}_{\mathbb{F}}\{w\}$, and a simple
orthogonal decomposition is
\begin{equation*}
    W = K \oplus \rad(W), \qquad K = \mathrm{span}_{\mathbb{F}}\{u,v\}.
\end{equation*}
Hence $\irk(K) = \frac{1}{2}\dim_{\mathbb{F}}(K) = 1$, and therefore
\begin{equation*}
    \irk(W) = \frac{1}{2}\dim_{\mathbb{F}}(K) + \dim_{\mathbb{F}}(\rad(W)) = 2.
\end{equation*}
Clearly, $\mathrm{span}_{\mathbb{F}}\{v,w\}$ is an example of an isotropic subspace of $W$ which, 
having vector space dimension $2$, is maximal.
\end{example}

\begin{lemma}\label{lem:splittingWperp}
    Let $W \leq V$, and let $W = \rad(W) \oplus K$ be an orthogonal splitting of $W$. Then $W^\perp$ orthogonally splits as $W^\perp = \rad(W) \oplus K'$, where $K'$ is symplectic. In particular, $\rad(W)^\perp = W + W^\perp = \rad(W) \oplus K \oplus K'$.
\end{lemma}
\begin{proof}
   We have $\rad(W^\perp) = W^\perp \cap (W^\perp)^\perp = \rad(W)$ and the statement easily follows.
\end{proof}

{By Remark~\ref{rem:dimirk}, $\irk(W)$ is the vector space dimension of an isotropic subspace of $V$, and hence $\irk(W)\le\irk(V)=\dim(V)$. Therefore, we have $\dim(W) \le \irk(W) \le \dim(V)$. Any symplectic subspace $W$ of $V$ achieves the lower bound $\dim(W) = \irk(W)$. Saturating the upper bound defines a class of
subspaces critical to our study.

\begin{definition}\label{def:stabsub}
    We say that $W\leq V$ is a \textbf{stabilizer} subspace if $\irk(W)=\dim(V)$.
\end{definition}

One can observe that $W$ is a stabilizer subspace if it contains a maximal isotropic subspace
of~$V$. As one would expect, quantum stabilizer codes correspond to stabilizer subspaces, where the
stabilizer of the code is precisely the radical of the associated subspace.

\begin{example}\label{example:2x2_bacon_shor1}
    Let $W = \mathrm{span}_\F\{r,s,t,u\} < V = \mathbb{F}^8$, with $\omega(r,s) = \omega(r,u) = \omega(t,s) = \omega(t,u) = 1$ and $\omega(r,t) = \omega(s,u) = 0$. In coordinates defined by $(r,s,t,u)$, we have
    $$\left[\left.\omega\right|_W\right] = \begin{pmatrix} 0 & 1 & 0 & 1\\ -1 & 0 & -1 & 0\\ 0 & 1 & 0 & 1\\ -1 & 0 & -1 & 0\end{pmatrix}.$$
    It is not immediately clear how such a $W$ can be realized. We have $\rad(W) = \mathrm{span}_\F\{r-t, s-u\}$, and the orthogonal decomposition $W = K\oplus \rad(W)$ with $K=\mathrm{span}_\F\{r,s\}$. Hence $\dim(W) = 1$ and $\irk(W) = 3$. However, $W$ is not a stabilizer subspace as $\dim(V) = 4$.
\end{example}

\begin{theorem}\label{thm:stab}
    For any $W \leq V$, we have $\rad(W) \leq W^\perp$, with equality if and only if $W$ is a stabilizer subspace.
\end{theorem}
\begin{proof}
   Clearly, we have $\rad(W) = W \cap W^\perp \leq W^\perp$. Let $n = \dim(V)$, $k = \dim(W)$, and $s = \irk(W)$. Then $\dim_\F(W^\perp) = 2n - k - s$ and $\dim_\F(\rad(W)) = s - k$. Therefore, $\dim_\F(W^\perp) - \dim_\F(\rad(W)) = 2n - k - s - (s - k) = 2(n - s)$. In particular, this quantity is zero if and only if $s = n$.
\end{proof}

As an immediate consequence, we have the following result.

\begin{corollary}\label{cor:stab}
   We have that $W \leq V$ is a stabilizer subspace if and only if $W^\perp$ is isotropic.
\end{corollary}

\begin{example}\label{example:repetition_code2}
    Consider the same setting as Example~\ref{example:repetition_code1}, that is $V = \mathbb{F}^4$ and $W = \mathrm{span}_{\mathbb{F}}\{u,v,w\}$, with $\omega(u,v) = 1$ and $\omega(u,w) = \omega(v,w) = 0$. We have $\dim(V) = 2 = \irk(W)$, and hence $W$ is a stabilizer subspace. Indeed, $\mathrm{span}_{\mathbb{F}}\{v,w\}$ is a maximal isotropic subspace of $W$, which is also maximal isotropic in $V$. Moreover, since $\omega(u,v) \neq 0$, one sees that~$W^\perp = \mathrm{span}_{\mathbb{F}}\{w\}$, which coincides with $\rad(W)$.
\end{example}

The following result establishes relations between the dimension and isorank of any subspace of $V$ and those of its dual.

\begin{proposition}\label{thm:dim-irk-orthogonality}
    For any $W\leq V$, we have
    \begin{equation*}
        \dim(W^\perp)= \dim(V) - \irk(W)\qquad\textup{ and }\qquad\irk(W^\perp)= \dim(V) - \dim(W).
    \end{equation*}
\end{proposition}
\begin{proof}
By Lemma~\ref{lem:splittingWperp}, $\dim_\F(\rad(W)^\perp) = \dim_\F(\rad(W)) + \dim_\F(K) + \dim_\F(K')$, for some symplectic subspaces $K, K' \leq \rad(W)^\perp$. This implies that
\begin{equation}
\label{eq:dimFV}
\dim_\F(V) = 2\dim_\F(\rad(W)) + \dim_\F(K) + \dim_\F(K')
\end{equation}
and we recall that the following hold:
\begin{enumerate}[label=(\roman*)]
    \item $\dim_\F(V) = 2\dim(V)$,
    \item $\dim(W^\perp) = \frac{1}{2} \dim_\F(K')$,
    \item $\irk(W) = \frac{1}{2} \dim_\F(K) + \dim_\F(\rad(W))$.
\end{enumerate}
By substituting these expressions into equation~\eqref{eq:dimFV}, we obtain
\begin{equation*}
    \dim_\F(V) = 2\dim(W^\perp) + \dim_\F(K) + 2\dim_\F(\rad(W)) = 2(\dim(W^\perp) + \irk(W)),
\end{equation*}
which proves the first part of the statement. The second part easily follows by replacing~$W$ with $W^\perp$.
\end{proof}

Note that, for any $W_1\leq W_2\leq V$, we have
\begin{equation}
    \label{eqn:dim-irk-decomp}
    \dim_\F(W_2)- \dim_\F(W_1)= (\dim (W_2) - \dim(W_1))+ (\irk(W_2) - \irk(W_1)).
\end{equation}

The following result establishes that both dimension and isorank are monotone with respect to subspace inclusion.

\begin{theorem}\label{thm:dim-irk-monotone}
    Let $W_1 \leq W_2 \leq V$ be subspaces. We have that $\dim(W_1) \leq \dim(W_2)$ and~$\irk(W_1) \leq \irk(W_2)$.
\end{theorem}
\begin{proof}
Since $W_1\leq W_2$, every set of pairwise orthogonal symplectic pairs in $W_1$ is also contained in $W_2$, and hence $\dim(W_1) \leq \dim(W_2)$. Similarly, every isotropic subspace of $W_1$ is also an isotropic subspace of $W_2$. Then as in Remark~\ref{rem:dimirk}, a maximal isotropic subspace of~$W_1$ can be extended to a maximal isotropic subspace of $W_2$ and hence $\irk(W_1) \leq \irk(W_2)$.
\end{proof}

Recall that the vector space dimension, when seen as a map, is monotone and modular. That is, it is increasing and satisfies
\begin{equation}\label{eqn:dimF-modularity}
    \dim_\F (W_1 + W_2) + \dim_\F (W_1 \cap W_2) = \dim_\F(W_1) + \dim_\F(W_2).
\end{equation}

Moreover, orthogonal subspaces $W_1,W_2\leq V$ may still intersect nontrivially. The following result shows that $\dim$ and $\irk$ are modular on pairs of orthogonal subspaces.

\begin{proposition}\label{prop:orth-modular}
     Let $W_1,W_2\leq V$ be orthogonal subspaces. Then we have
     \begin{equation}\label{eqn:dim-modular}
         \dim(W_1 + W_2) + \dim(W_1\cap W_2)= \dim(W_1) + \dim(W_2)
     \end{equation}
     and
     \begin{equation}\label{eqn:irk-modular}
         \irk(W_1 + W_2) + \irk(W_1\cap W_2)= \irk(W_1) + \irk(W_2).
     \end{equation}
\end{proposition}
\begin{proof}
   Let $W_1\cap W_2=K\oplus S$ be an orthogonal splitting, and $S=\spn\{e^{(1)},\ldots,e^{(s)}\}$ with~$s=\irk(W_1\cap W_2)$. As each $e^{(j)}\in W_2\subseteq W_1^\perp$ we have $e^{(j)}$ is isotropic in $W_1$ as well. This implies $S\leq \rad(W_1)$ and we can extend the basis of $S$ so that
   \begin{equation*}
    \rad(W_1) = \mathrm{span}\{e^{(1)}, \dots, e^{(s)}, e^{(s+1)}, \dots, e^{(s+d_1)}\},
    \end{equation*}
    which shows that $\irk(\rad(W_1))=s+d_1$. Analogously, we have 
    \begin{equation*}
        \rad(W_2) = \mathrm{span}\{e^{(1)}, \dots, e^{(s)}, \varepsilon^{(s+1)}, \dots, \varepsilon^{(s+d_2)}\}
    \end{equation*} 
    and hence $\irk(\rad(W_2)) = s + d_2$. We claim that
    \begin{equation}\label{eqn:isotropic-span}
        \{e^{(1)}, \dots, e^{(s)}, e^{(s+1)}, \dots, e^{(s+d_1)}, \varepsilon^{(s+1)}, \dots, \varepsilon^{(s+d_2)}\}
    \end{equation} 
    is a maximal independent set of isotropic vectors in $W_1 + W_2$, which implies $\rad(W_1+W_2)=\rad(W_1)+\rad(W_2)$ and hence $\dim_{\mathbb F}(\rad(W_1+W_2))=s+d_1+d_2$. First, notice that the set in~\eqref{eqn:isotropic-span} is independent by construction. Then, each $e^{(j)}$ is isotropic in $W_1$ and, since~$W_1\subseteq W_2^\perp$, it is orthogonal to $W_2$. It follows that $e^{(j)}$ is isotropic in $W_1+W_2$. Analogously, each $\varepsilon^{(j)}$ is isotropic in $W_1+W_2$, and therefore the set in~\eqref{eqn:isotropic-span} is isotropic. Finally, note that if $w_1\in W_1$ and $w_2\in W_2$ are such that $w_1+w_2$ is isotropic in $W_1+W_2$, then for every $v\in W_1$ we have
    \begin{equation*}
        0 = \omega(v,w_1+w_2)= \omega(v,w_1) + \omega(v,w_2)= \omega(v,w_1).
    \end{equation*}
    Hence $w_1$ is isotropic in $W_1$ and, analogously, $w_2$ is isotropic in $W_2$. This means that $w_1+w_2$ is in the span of the set in~\eqref{eqn:isotropic-span}, and hence such a set is maximal. 
    Since $W_1\cap W_2\leq W_2\leq W_1^\perp$, the subspace $W_1\cap W_2$ is isotropic, so that $K=\{0\}$, $W_1\cap W_2=S$, and $\dim(W_1\cap W_2)=0$. By~\eqref{eqn:dimF-modularity}, we have $\dim_\F(W_1+W_2)=\dim_\F(W_1)+\dim_\F(W_2)-s$. Since $2\dim(W)=\dim_\F(W)-\dim_\F(\rad(W))$ for every $W\leq V$, we obtain
    \begin{equation*}
    2\dim(W_1+W_2)=\dim_\F(W_1)+\dim_\F(W_2)-s-(s+d_1+d_2)=2\dim(W_1)+2\dim(W_2),
    \end{equation*}
    which proves~\eqref{eqn:dim-modular}. Finally, by~\eqref{eqn:dim-irk-decomp} applied to $\{0\}\leq W$, we have $\irk(W)=\dim_\F(W)-\dim(W)$ for every $W\leq V$, and~\eqref{eqn:irk-modular} follows by subtracting~\eqref{eqn:dim-modular} from~\eqref{eqn:dimF-modularity}.
\end{proof}

\section{Error Correcting Codes}

In this section, we introduce the notion of symplectic codes and characterize their fundamental coding-theoretic parameters. We further define the analogue of an anticode in the symplectic setting and demonstrate how this framework aligns with the theories developed in~\cite{ravagnani2016generalized} and~\cite{byrne2023tensor}. Moreover, we show that anticodes naturally give rise to code operations such as \textit{puncturing} and \textit{shortening}, and that these operations provide a way to translate and generalize results and properties of quantum codes within an algebraic-combinatorial framework. In the remainder, we assume that $V$ is a symplectic space with $\dim(V)=1$, that is $V=\spn\{e,f\}$ with $\omega(e,f)=1$. 

\subsection{Codes in a Symplectic Space}

We define the symplectic space $V^n = V \times \cdots \times V$, and notice that $\dim(V^n)=n$. For an element $v\in V^n$, we write $v=(v_1,\ldots,v_n)$, with $v_j\in V$.

\begin{definition}
    The \emph{Hamming weight} $\wt(v)$ of $v\in V^n$ is the number of nonzero components of $v$, that is $\wt(v)=|\{j\in\{1,\ldots, n\}:v_j\neq 0\}|$.
\end{definition}

It is worth noting that other weight functions could also be of interest in this setting. For instance, when $\F=\F_2$, one might consider the \textit{complete weight} of a vector $v\in V^n$, defined as $\overline{\wt}(v) = (\wt_0(v), \wt_e(v), \wt_f(v), \wt_{e+f}(v))$, where $\wt_x(v)$ denotes the number of components of $v$ equal to $x$. However, in this work, we restrict our attention to the Hamming weight.

\begin{definition}\label{def:sympcodes}
    A (\textbf{symplectic}) \textbf{code} is a linear subspace $C \leq V^n$. The \textbf{length} of $C$ is~$n=\dim(V^n)$, its \textbf{dimension} is $k=\dim(C)$, and its \textbf{isorank} is $s=\irk(C)$. The \textbf{minimum distance} of $C$ is $d = \min\{\wt(v) : v \in C \setminus \rad(C)\}$. The \textbf{maximum weight} of~$C$ is $\maxwt(C)=\max\{\wt(v):v\in C\}$.
\end{definition}

When $\F=\F_q$ is a finite field of cardinality $q$, we use the notation $[[n,k,s,d]]_q$ to denote a symplectic code $C$ of length $n$, dimension $k$, isorank $s$, and minimum distance $d$. The following definition is motivated by Definition~\ref{def:stabsub}, Theorem~\ref{thm:stab}, and Corollary~\ref{cor:stab}.

\begin{definition}\label{def:stabcode}
    A \textbf{stabilizer code} is a stabilizer subspace  $C \leq V^n$. In particular, its radical satisfies $\rad(C)=C^\perp$ and $C^\perp$ is isotropic. 
\end{definition}

Vectors that are isotropic in $C$, that is, elements of $\rad(C)$, play a distinct role in these codes: their contribution is reflected in the isorank, and they do not affect the minimum distance of the code.

\begin{remark}\label{rem:stabilizer}
   One can easily verify that this definition recovers the notion of stabilizer codes introduced by Gottesman~\cite{gottesman1997stabilizer}. On the other hand, from a coding-theoretic perspective, one can check that stabilizer codes are precisely dual-containing symplectic codes. In this setting, the assumption that $V$ is spanned by a symplectic pair $(e,f)$ is necessary, since $e$ and $f$ are naturally associated with the Pauli operators $X$ and $Z$ (see also Example~\ref{example:repetition_code3}). 
\end{remark}

\begin{example}\label{example:repetition_code3}
    Let $\mathfrak{H} = (\mathbb{C}^2)^{\otimes 2}$ be a Hilbert space, and let $\mathfrak{C}$ denote the quantum repetition code. That is, $\mathfrak{C}$ is a stabilizer code with stabilizer $\mathcal{S}(\mathfrak{C}) = \langle Z \otimes Z \rangle$, and normalizer $
    \mathcal{N}(\mathfrak{C}) = \langle Z \otimes Z,\, X \otimes X,\, Z \otimes I \rangle$. The logical operators are $\bar{X} = X \otimes X$, 
    $\bar{Z} = Z \otimes I$. This is a~$[[2,1,1]]_2$ quantum code. We identify  $X \leftrightarrow e$, $Z \leftrightarrow f$, $\mathcal{N}(\mathfrak{C}) \leftrightarrow C$ and $\mathcal{S}(\mathfrak{C}) \leftrightarrow C^\perp$. We get
    \begin{equation*}
        C=\operatorname{span}_{\F_2}\{(e,e),(f,f),(f,0)\}\qquad\textup{ and }\qquad C^\perp=\operatorname{span}_{\F_2}\{(f,f)\}.
    \end{equation*}
Thus, the quantum repetition code is a realization of the subspace described in  Examples~\ref{example:repetition_code1} and~\ref{example:repetition_code2},  with underlying field $\mathbb{F}_2$. 
The length, dimension, and minimum distance of $C$ coincide with those of the quantum code $\mathfrak{C}$, its isorank is $2$, and we have $C^{\perp} = \rad(C)$, in accordance with Remark~\ref{rem:stabilizer} and Definition~\ref{def:stabcode}.
\end{example}

\begin{remark}\label{rem:subsys}
It is interesting to observe that the symplectic setting provides a general and flexible framework for studying quantum error-correcting codes. 
A further example is given by \textbf{subsystem codes}, introduced by Poulin~\cite{poulin2005stabilizer}, which naturally arise from a pair of symplectic codes $C$ and $D$ satisfying $C^{\perp}\le D\le C$ and $C^{\perp}=\rad(D)$. Here~$C^{\perp}$, $D$ and~$C$ identify the stabilizer, gauge group and normalizer respectively, and the parameters of the associated $[[n,k,r,d]]_q$ subsystem code are recovered as $r=\dim(D)$ and $k=\dim(C)-\dim(D)$. Note that these are attached to the pair $(C,D)$ rather than to $C$ alone, so that $\dim(C)=k+r$. 
In particular, the gauge group $D$, viewed as a symplectic code, has dimension $r$ and isorank $n-k$. Note also that the distance of the subsystem code is $d=\min\{wt(v): v \in C \setminus D\}$, which in general differs from the minimum distance of $C$ in the sense of Definition~\ref{def:sympcodes}. We illustrate this correspondence in Example~\ref{example:2x2_bacon_shor_new}, where we consider the $2\times~2$ Bacon-Shor code, a family of subsystem codes introduced by Bacon in~\cite{bacon2006operator}.  
Similarly, \textbf{entanglement-assisted quantum codes}~\cite{brun2006correcting} arise by dropping the requirement that $C$ be a stabilizer subspace: here $C^{\perp}$ is the (not necessarily self-orthogonal) stabilizer, and the parameters of the associated $[[n,k,d;c]]_q$ code are
\begin{equation*}
    c=n-\irk(C)\qquad\text{and}\qquad k=\dim(C),
\end{equation*}
by Proposition~\ref{thm:dim-irk-orthogonality}, since $C^\perp$ has $c$ independent pairwise orthogonal symplectic pairs and $\irk(C^\perp)=n-k$. Moreover, we have
$\rad(C)=C\cap C^{\perp}$ and hence $C\setminus\rad(C)=C\setminus C^{\perp}$,
so that the minimum distance of Definition~\ref{def:sympcodes} agrees with the
usual one. The case of stabilizer codes is recovered when $\irk(C)=n$, which gives $c=0$. Note that the subsystem notation $[[n,k,r,d]]_q$, whose third entry is the number of gauge qubits, differs from the notation of Definition~\ref{def:sympcodes}, whose third entry is the isorank. In the latter notation, an entanglement-assisted $[[n,k,d;c]]_q$ code corresponds to a symplectic $[[n,k,n-c,d]]_q$ code.
\end{remark}

\begin{example}\label{example:2x2_bacon_shor_new}
    Let $\mathfrak{H} = (\mathbb{C}^2)^{\otimes 4}$ be a Hilbert space, and let $\mathfrak{C}$ denote the $2\times 2$ Bacon-Shor code. That is, $\mathfrak{C}$ is the subsystem code with gauge group  
\[
    \mathcal{G}(\mathfrak{C}) = \langle 
    X\otimes X\otimes I\otimes I,\,
    I\otimes I\otimes X\otimes X,\,
    Z\otimes I\otimes Z\otimes I,\,
    I\otimes Z\otimes I\otimes Z
    \rangle,
\]
and stabilizer  $\mathcal{S}(\mathfrak{C}) = \langle 
    X\otimes X\otimes X\otimes X,\,
    Z\otimes Z\otimes Z\otimes Z
    \rangle$. The logical operators are  $\bar{X} = X\otimes I\otimes X\otimes I$ and $ 
    \bar{Z} = Z\otimes Z\otimes I\otimes I$. This is a $[[4,1,2]]_2$ quantum code. We identify $X \leftrightarrow e$, $Z \leftrightarrow f$, $\mathcal{G}(\mathfrak{C}) \leftrightarrow D$, and $\mathcal{S}(\mathfrak{C}) \leftrightarrow C^{\perp}$. Hence, we obtain  
\[
    D = \mathrm{span}_{\mathbb{F}_2}\{(e,e,0,0), (0,0,e,e), (f,0,f,0), (0,f,0,f)\},
\]
 $C^{\perp} = \rad(D) = \mathrm{span}_{\mathbb{F}_2}\{(e,e,e,e), (f,f,f,f)\}$ and
\[
    C = \mathrm{span}_{\mathbb{F}_2}\{(e,e,0,0), (0,0,e,e), (f,f,0,0), (0,0,f,f), (e,0,e,0), (f,0,f,0)\}.
\]
Finally, one can readily verify that $C^{\perp} \subsetneq D \subsetneq C$, in line with Remark~\ref{rem:subsys}.
\end{example}

Next we establish a relationship between the dimension and the maximum weight of a code in a symplectic space (cf.~\cite[Proposition~6]{ravagnani2016generalized}). We recall that, in this context, the dimension refers to the maximum number of pairwise orthogonal symplectic pairs (see Remark~\ref{rem:dimirk}) rather than the vector space dimension.
We begin with a simple technical lemma, analogous to Gaussian reduction for codes over a field. 

\begin{lemma}\label{lemma:gaussian-reduce}
    Let $m=\dim_\F(C)$. There exist a basis $\{u^{(1)},\ldots,u^{(m)}\}$ of $C$ and a set $J\subseteq\{1,\ldots,n\}$ with $|J|\geq m/2$ such that the following hold.
    \begin{enumerate}
        \item The vector $u^{(1)}+\cdots+u^{(m)}$ has a nonzero component in each coordinate indexed by~$J$.
        \item If $m=2k$ and $|J|=k$, write $J=\{i_1,\ldots,i_k\}$. Then the basis consists of vectors $\bar v^{(1)},\bar w^{(1)},\ldots,\bar v^{(k)},\bar w^{(k)}$ such that, for each $j\in\{1,\ldots,k\}$, we have $\bar v^{(j)}_{i_j}=e$, $\bar w^{(j)}_{i_j}=f$, and $\bar v^{(\ell)}_{i_j}=\bar w^{(\ell)}_{i_j}=0$ for each $\ell\neq j$.
    \end{enumerate}
\end{lemma}
\begin{proof}
    For $v\in V^n$, write $v_j=a_je+b_jf$ with $a_j,b_j\in\F$ for each $j\in\{1,\ldots,n\}$, and call $a_j$ and $b_j$ the entries $2j-1$ and $2j$ of $v$. Let $M$ be the $m\times 2n$ matrix whose rows are the entries of the elements of a basis of $C$. By Gaussian elimination, we may assume that $M$ is in reduced row echelon form, and we let $u^{(1)},\ldots,u^{(m)}$ be the elements of $C$ corresponding to its rows. Then there exist pivot positions $t_1<\cdots<t_m$ in $\{1,\ldots,2n\}$ such that the~$t_r$-th entry of $u^{(\ell)}$ is $1$ if $\ell=r$, and $0$ otherwise. Let $J$ be the set of coordinates $j$ such that~$2j-1$ or~$2j$ is a pivot position. Since each coordinate corresponds to two entries, we have $|J|\geq m/2$. Moreover, for every $r\in\{1,\ldots,m\}$, the $t_r$-th entry of $u^{(1)}+\cdots+u^{(m)}$ is equal to $1$, which proves~(1). Suppose now that $m=2k$ and $|J|=k$. Then, for each $j\in J$, both~$2j-1$ and~$2j$ are pivot positions. Writing $J=\{i_1,\ldots,i_k\}$, we let $\bar v^{(j)}$ and $\bar w^{(j)}$ be the elements of the basis whose pivot positions are $2i_j-1$ and $2i_j$, respectively. Since each pivot column of $M$ has a single nonzero entry, we obtain~(2).
\end{proof}

The following is an immediate consequence of this lemma.

\begin{corollary}\label{corollary:anticode-bound}
    We have $k\leq \maxwt(C)$.
\end{corollary}
\begin{proof}
    Let $m=\dim_\F(C)$, and let $u^{(1)},\ldots,u^{(m)}$ and $J$ be as in Lemma~\ref{lemma:gaussian-reduce}. Since $\irk(C)\geq\dim(C)=k$ by Definition~\ref{def:dim-irk}, equation~\eqref{eqn:dim-irk-decomp} applied to $\{0\}\leq C$ gives $m=k+\irk(C)\geq 2k$. By Lemma~\ref{lemma:gaussian-reduce}, part~1, the vector $u^{(1)}+\cdots+u^{(m)}\in C$ has a nonzero component in each coordinate indexed by $J$, and therefore $\maxwt(C)\geq|J|\geq m/2\geq k$.
\end{proof}

\subsection{Anticodes in a Symplectic Space}

In this subsection, we introduce and investigate the notion of anticodes in symplectic spaces, inspired by the work in~\cite{ravagnani2016generalized} and~\cite{byrne2023tensor}.

\begin{definition}
    We say that $C$ is a (\textbf{symplectic}) \textbf{anticode} if it attains the bound in Corollary~\ref{corollary:anticode-bound}, that is if $k=\maxwt(C)$. We denote the set of all anticodes in $V^n$ by~$\mA(n)$.
\end{definition}

Let $J \subseteq \{1, \dots, n\}$. A \textbf{free code} in $V^n$ \textbf{supported} on $J$ is defined as
\begin{equation}
\{v \in V^n : v_j = 0 \textup{ for all } j \notin J\}.
\end{equation}
It is immediate that for any such free code $D$, we have $\dim(D) = \maxwt(D) = |J|$. Therefore, every free code is an anticode. The next result shows that the converse also holds.

\begin{proposition}\label{prop:anticode-char}
    A code $C$ is an anticode if and only if $C$ is a free code.
\end{proposition}
\begin{proof}
    It is immediate to check that every free code is an anticode, so it remains to prove the converse. Let $A$ be an anticode with $\dim(A)=\maxwt(A)=k$, and let $m=\dim_\F(A)$. Since $\irk(A)\geq\dim(A)=k$ by Definition~\ref{def:dim-irk}, equation~\eqref{eqn:dim-irk-decomp} applied to $\{0\}\leq A$ gives $m=k+\irk(A)\geq 2k$. Let $J$ be as in Lemma~\ref{lemma:gaussian-reduce} applied to $A$. By part~1 of the lemma, $A$ contains a vector with a nonzero component in each coordinate indexed by $J$, and hence
    \[
        k\leq m/2\leq|J|\leq\maxwt(A)=k.
    \]
    Therefore $m=2k$ and $|J|=k$. By part~2 of the lemma, writing $J=\{i_1,\ldots,i_k\}$, there exist $2k$ linearly independent vectors $\bar v^{(1)},\bar w^{(1)},\ldots,\bar v^{(k)},\bar w^{(k)}$ in $A$ such that, for each $j$, the components $\bar v^{(j)}_{i_j}$ and $\bar w^{(j)}_{i_j}$ form a basis of $V$, while $\bar v^{(\ell)}_{i_j}=\bar w^{(\ell)}_{i_j}=0$ whenever $\ell\neq j$.
    We claim that $A$ is supported on $J$. Let $D$ be the free code supported on $J$, and let $u\in A$. Since $V$ has dimension two over $\F$, we can choose, for each $j$, an element $t_j\in V$ such that $t_j\neq 0$ and $t_j+u_{i_j}\neq 0$. By the properties of the vectors above, there exists a vector
    \[
        v\in\spn\{\bar v^{(1)},\bar w^{(1)},\ldots,\bar v^{(k)},\bar w^{(k)}\}
    \]
    such that $v_{i_j}=t_j$ for every $j$. Consequently, both $v$ and $v+u$ are nonzero in all coordinates of $J$. Since $\maxwt(A)=k=|J|$, neither $v$ nor $v+u$ can have a nonzero coordinate outside~$J$. Therefore, for every $r\notin J$,
    \[
        u_r=(v+u)_r-v_r=0.
    \]
    Since $u$ was an arbitrary vector, we obtain $A\leq D$. Finally, the $2k$ linearly independent vectors $\bar v^{(1)},\bar w^{(1)},\ldots,\bar v^{(k)},\bar w^{(k)}$ are elements of $A$ and hence of $D$. Since $\dim_\F(D)=2k$, they form a basis of $D$. Therefore, $D\leq A$, and consequently $A=D$, proving that $A$ is a free code.
\end{proof}

\begin{remark}\label{rem:anticodes-symplectic}
It follows that every anticode $A$ is a symplectic subspace of $V^n$, that is
$\rad(A) = 0$. In particular, any $\dim(A)$ pairwise orthogonal symplectic pairs in $A$ form a basis of $A$. 
\end{remark}

As a consequence, every anticode is uniquely determined by its support. In what follows, for any subset $J \subseteq \{1, \dots, n\}$, we denote by $A_J$ the anticode supported on~$J$. It is interesting to observe that the set of anticodes $\mA(n)$ inherits a Boolean algebra structure from the power set of $\{1,\ldots,n\}$. In particular, the meet of anticodes $A_1, A_2 \in \mA(n)$ is given by $A_1 \cap A_2$, while their join is $A_1 + A_2$. The smallest and largest anticodes are~$\{0\}$ and $V^n$, respectively. Finally, the following is easy to check.

\begin{proposition}\label{prop:ant-compl}
Let $J \subseteq \{1, \dots, n\}$. We have $A_J^\perp=A_{J^c}$, where $J^c$ denotes the complement of $J$ in $\{1,\ldots,n\}$.
\end{proposition}

As we will see later, studying the intersections of a code with anticodes provides a powerful tool to capture its combinatorial and structural properties. We begin by establishing the following duality result, which is the analogue of~\cite[Lemma~28]{ravagnani2016rank} and~\cite[Lemma~6.5]{byrne2023tensor}, and  can be interpreted as a form of MacWilliams identity.

\begin{theorem}\label{thm:duality-distance}
    For any $A\in\mA(n)$, we have
    \begin{equation}\label{eqn:MacWilliams}
        \dim_\mathbb{F} (A \cap C)=
        \dim_\mathbb{F} (C) - \dim_\mathbb{F} (A^\perp) + \dim_\mathbb{F} (A^\perp \cap C^\perp).
    \end{equation}
    Moreover, if $C$ is a stabilizer code then
    \begin{equation}\label{eqn:MacWilliams-stabilizer}
         \dim (A^\perp) - \dim(A^\perp\cap C) - \irk(A^\perp\cap C)= \dim (A) - \irk(A\cap \rad(C)) - \dim (C).
    \end{equation}
\end{theorem}
\begin{proof}
    Recall that, by~\eqref{eqn:dimF-modularity}, we have
    \begin{equation}\label{eqn:MacWilliams-proof1}
    \dim_\mathbb{F}(A^\perp \cap C^\perp) + \dim_\mathbb{F}(A^\perp + C^\perp) = \dim_\mathbb{F} (A^\perp) + \dim_\mathbb{F} (C^\perp).
    \end{equation}
    Additionally, since $\omega$ is nondegenerate on $V^n$, we have $\dim_\mathbb{F}( C^\perp) = \dim_\mathbb{F} (V^n) - \dim_\mathbb{F} (C)$ and
    \begin{equation*}
        \dim_\mathbb{F}(A^\perp + C^\perp) = \dim_\mathbb{F} (V^n) - \dim_\mathbb{F}((A^\perp + C^\perp)^\perp)=\dim_\mathbb{F} (V^n) - \dim_\mathbb{F}(A \cap C).
    \end{equation*}
    Combining these with~\eqref{eqn:MacWilliams-proof1} establishes~\eqref{eqn:MacWilliams}. Suppose now that $C$ is a stabilizer code, so that $\rad(C)=C^\perp$ and $\irk(C)=n$. By Proposition~\ref{prop:ant-compl}, $A^\perp$ is an anticode with $(A^\perp)^\perp=A$, and~\eqref{eqn:MacWilliams} applied to $A^\perp$ gives
    \begin{equation*}
        \dim_\F(A^\perp\cap C)=\dim_\F(C)-\dim_\F(A)+\dim_\F(A\cap\rad(C)).
    \end{equation*}
    By~\eqref{eqn:dim-irk-decomp} applied to $\{0\}\leq W$, we have $\dim_\F(W)=\dim(W)+\irk(W)$ for every $W\leq V^n$. In particular, $\dim_\F(C)=\dim(C)+n$, $\dim_\F(A)=2\dim(A)$ since $A$ is symplectic, and $\dim_\F(A\cap\rad(C))=\irk(A\cap\rad(C))$ since $A\cap\rad(C)$ is isotropic. Substituting these expressions and using $n=\dim(A)+\dim(A^\perp)$, we obtain~\eqref{eqn:MacWilliams-stabilizer}.
\end{proof}

\subsection{Duality of Puncturing and Shortening}

In this subsection, let $A$ be the anticode supported on $J$, for some 
$J \subseteq \{1, \ldots, n\}$. Consider the projection map 
$\pi_J : V^n \to V^{|J|}$, which projects onto the coordinates indexed by $J$.  It is straightforward to verify that $V^{|J|}$ is isomorphic to the anticode $A$. Consequently, we may view an anticode either as a subspace of $V^n$ or as an independent symplectic space, with the symplectic form and basis naturally inherited from $V^n$.

\begin{definition}
    We define the \emph{puncturing} of $C$ on $A$ as $\Pi_A C = \pi_J(C)$. We define the \emph{shortening} of $C$ on $A$ as $\Sigma_A C = \pi_J(C \cap A)$.
\end{definition}

The shortening can be seen as the puncturing of the \textit{restriction} of $C$ to $A$. We recall the following fundamental result in quantum error correction \cite[Lemma 1]{bravyi2009no}.

\begin{theorem}[Cleaning Lemma]
    Let $C$ be a stabilizer code with radical $C^\perp$. Then one of the following holds.
    \begin{enumerate}
        \item There exists an element of $(A\cap C) \setminus (A\cap C^\perp)$.
        \item For any $c \in C$ there exists an $s \in C^\perp$ so that $\pi_J(c) = \pi_J(s)$.
    \end{enumerate}
\end{theorem}

The use of ``cleaning'' in this result is based on case (2) of the theorem: given a $c \in C$ we can shift it by an element of the radical $s\in C^\perp$ so that the support of $c-s$ is cleaned off of the indices $J$. The exceptional case (1) is that $\mathrm{supp}(c)\subseteq J$ already and $c \not\in C^\perp$. We claim that the cleaning lemma is just a special case of a much more general duality relation between the operations of puncturing and shortening, namely the analogue of~\cite[Theorem~1.5.7(1)]{huffman2003fundamentals} in the symplectic setting.  Hence we refer to this result as the \textit{Generalized Cleaning Lemma}.

\begin{theorem}[Generalized Cleaning Lemma]\label{thm:cleaning-lemma}
    For any code $C \leq V^n$ and any anticode $A \in \mA(n)$, we have  
    \begin{equation*}
        \Sigma_A C^\perp = (\Pi_A C)^\perp\quad\textup{ and }\quad \Pi_A C^\perp = (\Sigma_A C)^\perp,
    \end{equation*}
    where the orthogonal complements on the right-hand sides are taken in $V^{|J|}$.
\end{theorem}
\begin{proof}
Without loss of generality, let $J = \{1, \dots, j\}$ for some $j\in\{1,\dots,n\}$, and let $A = A_J$. Then we have the decomposition $V^n = A \oplus A^{\perp}$. By Proposition~\ref{prop:ant-compl}, every element $v\in V^n$ can be written as $v = (v_1, v_2)$ with $(v_1, 0)\in A$ and $(0, v_2)\in A^{\perp}$.
Let $u\in\Sigma_A C^{\perp}$, so that $(u,0)\in C^{\perp}$. Then for any $(v,w)\in C$ we have
\begin{equation*}
    \omega|_A(u, v) = \omega((u,0), (v,w)) = 0.
\end{equation*}
Since every element of $\Pi_A C$ is of the form $v$ for some $(v,w)\in C$, this shows $\Sigma_A C^{\perp}\subseteq(\Pi_A C)^{\perp}$.
Conversely, let $u\in(\Pi_A C)^{\perp}$ and let $(v,w)\in C$. Then $v\in\Pi_A C$, and hence
\begin{equation*}
    \omega((u,0),(v,w)) = \omega|_A(u,v) = 0.
\end{equation*}
Since $(v,w)\in C$ was arbitrary, we conclude that $(u,0)\in C^{\perp}$, that is $u\in\Sigma_A C^{\perp}$. This shows $(\Pi_A C)^{\perp}\subseteq\Sigma_A C^{\perp}$ and thus $\Sigma_A C^{\perp} = (\Pi_A C)^{\perp}$. The second equality easily follows by replacing $C$ with $C^{\perp}$.
\end{proof}

We will often use the next result, which is an immediate consequence of the theorem above. Recall that $d$ denotes the minimum distance of the code $C$.

\begin{corollary}\label{corollary:cleaning}
    Let $C$ be a stabilizer code (i.e., $\rad(C)=C^\perp$) and let $A \in \mA(n)$. If $\dim(A) < d$ then $\Pi_A C = \Pi_A \rad(C)$.
\end{corollary}

\subsection{Complementarity}

In this subsection, we use Theorem~\ref{thm:cleaning-lemma} to compare the two sides of a bipartition $\{J,J^c\}$ of the coordinates. In quantum error correction, this comparison is captured by complementary recovery, where complementary subsystems can reconstruct commuting code subalgebras whose union is as large as possible~\cite{harlow2017ryu,pollack2022understanding}. The results below provide an algebraic counterpart of this phenomenon. The part of the radical that is contained neither in $A$ nor in $A^\perp$ projects onto isomorphic spaces on the two sides (Proposition~\ref{prop:dim-irk-S'}), and this yields identities relating the shortenings of a stabilizer code on $A$ and on $A^\perp$ (Corollary~\ref{cor:stab-complementarity}).

As in the previous subsection, let $J \subseteq \{1, \ldots, n\}$ and let $A$ be the anticode supported on~$J$.  Recall that $V^n$ decomposes orthogonally as $V^n = A \oplus A^\perp$. We consider the orthogonal splitting
\begin{equation*}
    \rad(C) = (\operatorname{rad}(C) \cap A) \oplus (\rad(C) \cap A^\perp) \oplus S'
\end{equation*}
for some $S' \subseteq V^n$. We also recall that all of these spaces are isotropic.

\begin{lemma}\label{lem:injective}
    The map $\pi_J:S'\rightarrow A$ is injective.
\end{lemma}
\begin{proof}
    It is straightforward to verify that, for any $v \in \rad(C)$, we have $\pi_J(v) = 0$ if and only if $v \in \rad(C) \cap A^\perp$. Therefore, for any $v \in S'$, we have $\pi_J(v) = 0$ if and only if $v = 0$. This concludes the proof.
\end{proof}

This lemma implies $\Pi_A\rad(C) = \Sigma_A\rad(C) \oplus \Pi_A S'$
and $\Pi_{A^\perp}\rad(C) = \Sigma_{A^\perp}\rad(C) \oplus \Pi_{A^\perp}S'$.
Moreover, the puncturing of $S'$ has complementarity properties, as shown in the next result.

\begin{proposition}\label{prop:dim-irk-S'}
   We have $\Pi_A S' \cong \Pi_{A^\perp} S'$, and in particular
\begin{equation*}
    \dim(\Pi_A S') = \dim(\Pi_{A^\perp} S') 
    \quad \text{and} \quad 
    \irk(\Pi_A S') = \irk(\Pi_{A^\perp} S').
\end{equation*}
\end{proposition}
\begin{proof}
    By Lemma~\ref{lem:injective}, the restriction $\pi_J|_{S'}:S'\to\Pi_A S'$ is injective. Since $\Pi_A S'=\pi_J(S')$ by definition, it is also surjective, and therefore an $\F$-linear isomorphism. Similarly, since $S'\cap A=\{0\}$ by the decomposition of $\rad(C)$, the restriction $\pi_{J^c}|_{S'}:S'\to\Pi_{A^\perp}S'$ is injective and hence an $\F$-linear isomorphism onto its image. Consequently, there is a natural
    $\F$-linear isomorphism
    \[
        T:\Pi_A S'\longrightarrow
        \Pi_{A^\perp}S' \text{ by }\pi_J(v) \longmapsto \pi_{J^c}(v).
    \]
    Since $S'\subseteq\rad(C)$ is isotropic,
    for any $u,v\in S'$ we have
    \[
        0=\omega(u,v)
        =\omega|_A(\pi_J(u),\pi_J(v))
        +\omega|_{A^\perp}
        (\pi_{J^c}(u),\pi_{J^c}(v)).
    \]
    Therefore, we get
    \[
        \omega|_{A^\perp}(T(x),T(y))
        =-\omega|_A(x,y)
    \]
    for all $x,y\in\Pi_A S'$.
    Thus, $T$ is an anti-isometry between
    the restricted alternating forms.
    Consequently, the two spaces have the same $\F$-dimension and their radicals have the same $\F$-dimension. By Definition~\ref{def:dim-irk}, we conclude that
    \[
        \dim(\Pi_A S')
        =\dim(\Pi_{A^\perp}S') \text{ and }
    \irk(\Pi_A S')
        =\irk(\Pi_{A^\perp}S').
    \]
   Finally, since every alternating bilinear form is isometric to its negative, the two spaces are isomorphic as vector spaces equipped with alternating forms, proving the claim.
\end{proof}

\begin{example}
   Let $\mathfrak{H} = (\mathbb{C}^2)^{\otimes 9}$ be a Hilbert space, and let $\mathfrak{C}$ denote the $[[9,1,3]]_2$ Shor code. That is, $\mathfrak{C}$ is the stabilizer code with stabilizer $\mathcal{S}(\mathfrak{C})=\langle s_1,\ldots,s_8\rangle$, where  
\begin{equation*}
\begin{array}{ll}
    s_1=Z\otimes Z \otimes I \otimes I\otimes I\otimes I\otimes I\otimes I\otimes I, & s_2= I\otimes Z \otimes Z \otimes I\otimes I\otimes I\otimes I\otimes I\otimes I, \\[1ex]
    s_3= I\otimes I \otimes I \otimes Z\otimes Z\otimes I\otimes I\otimes I\otimes I, & s_4=I\otimes I \otimes I \otimes I\otimes Z\otimes Z\otimes I\otimes I\otimes I,\\[1ex]
    s_5= I\otimes I \otimes I \otimes I\otimes I\otimes I\otimes Z\otimes Z\otimes I, & s_6=I\otimes I \otimes I \otimes I\otimes I\otimes I\otimes I\otimes Z\otimes Z,\\[1ex]
    s_7= X\otimes X \otimes X \otimes X\otimes X\otimes X\otimes I\otimes I\otimes I, & s_8=I\otimes I \otimes I \otimes X\otimes X\otimes X\otimes X\otimes X\otimes X.
\end{array}
\end{equation*}

Identifying $X \leftrightarrow e$, $Z \leftrightarrow f$, $\mathcal{S}(\mathfrak{C}) \leftrightarrow C^\perp$, and $\mathcal{N}(\mathfrak{C}) \leftrightarrow C$, we obtain that $s_i\leftrightarrow c_i$ for all~$i\in\{1,\ldots,8\}$, where
\begin{equation*}
\begin{array}{lll}
    c_1=(f,f,0,0,0,0,0,0,0), & &c_2=(0,f,f,0,0,0,0,0,0),\\[1ex]
    c_3=(0,0,0,f,f,0,0,0,0), && c_4=(0,0,0,0,f,f,0,0,0),\\[1ex]
    c_5=(0,0,0,0,0,0,f,f,0), && c_6=(0,0,0,0,0,0,0,f,f),\\[1ex]
    c_7=(e,e,e,e,e,e,0,0,0), && c_8=(0,0,0,e,e,e,e,e,e).\\[1ex]
\end{array}
\end{equation*}

Hence $\rad(C)=C^\perp=\text{span}_{\F_2}\{c_1,\ldots,c_8\}$. Let $A=A_{\{1,2,3,4\}}$ be the anticode supported on~$\{1,2,3,4\}$, and recall that $A^\perp=A_{\{5,6,7,8,9\}}$ by Proposition~\ref{prop:ant-compl}. Then  
\begin{equation*}
    \rad(C)\cap A=\text{span}_{\F_2}\{c_1,c_2\}\qquad \textup{ and }\qquad
    \rad(C)\cap A^\perp=\text{span}_{\F_2}\{c_4,c_5,c_6\}.
\end{equation*}

This leads to the orthogonal decomposition $\rad(C)=(\rad(C)\cap A)\oplus(\rad(C)\cap A^\perp)\oplus S'$,
where~$S'=\text{span}_{\F_2}\{c_3,c_7,c_8\}$. Puncturing $S'$ on $A$ and $A^\perp$ yields  
\begin{align*}
    \Pi_A(S')&=\text{span}_{\F_2}\{(0,0,0,f),(e,e,e,e),(0,0,0,e)\},\\
    \Pi_{A^\perp}(S')&=\text{span}_{\F_2}\{(f,0,0,0,0),(e,e,0,0,0),(e,e,e,e,e)\}.
\end{align*}

One can easily verify that these spaces are isomorphic, and each contains a symplectic pair, namely $((0,0,0,f),(0,0,0,e))$ in $\Pi_A(S')$ and $((f,0,0,0,0),(e,e,0,0,0))$ in $\Pi_{A^\perp}(S')$. Moreover,  we have
\begin{equation*}
    \rad(\Pi_A(S'))=\text{span}_{\F_2}\{(e,e,e,0)\}\qquad\textup{ and }\qquad 
    \rad(\Pi_{A^\perp}(S'))=\text{span}_{\F_2}\{(0,0,e,e,e)\}.
\end{equation*}
This implies  $\dim(\Pi_A(S'))=\dim(\Pi_{A^\perp}(S'))=1$ and $
\irk(\Pi_A(S'))=\irk(\Pi_{A^\perp}(S'))=2$, in line with Proposition~\ref{prop:dim-irk-S'}.
\end{example}

We conclude this subsection with some consequences of this proposition.

\begin{corollary}
    For any $A\in\mA(n)$ we have
    \begin{equation}\label{eqn:complementarity-dim}
        \dim(\Pi_A \rad(C))= \dim(\Pi_{A^\perp} 
        \rad(C)),
    \end{equation}
    and
    \begin{equation}\label{eqn:complementarity-irk}
        \irk(\Pi_A \rad(C)) - \irk(\Sigma_A \rad(C))= \irk(\Pi_{A^\perp} \rad(C)) - \irk(\Sigma_{A^\perp} \rad(C)).
    \end{equation}
\end{corollary}
\begin{proof}
   Recall that, by the discussion above, we have $\Pi_A \rad(C) = \Sigma_A \rad(C) \oplus \Pi_A S'$, and similarly for $\Pi_{A^\perp} \rad(C)$. The two summands are orthogonal. Indeed, for any element $c\in\rad(C)\cap A$ and $s\in S'$, we have $\pi_{J^c}(c)=0$, and hence $\left.\omega\right|_A(\pi_J(c),\pi_J(s))=\omega(c,s)=0$, since $\rad(C)$ is isotropic. As $\Sigma_A \rad(C)$ is isotropic, Proposition~\ref{prop:orth-modular} gives $\dim (\Pi_A \rad(C)) = \dim(\Pi_A S')$ and $\irk(\Pi_A \rad(C)) = \irk(\Sigma_A \rad(C)) + \irk(\Pi_A S')$. The result now follows from Proposition~\ref{prop:dim-irk-S'}.
\end{proof}

\begin{corollary}\label{cor:stab-complementarity}
     Let $A\in\mA(n)$ and suppose that $C$ is a stabilizer code, that is $ \rad(C)=C^\perp$. We have
     \begin{equation}\label{eqn:stabilizer-complementarity-dim}
            \dim(A) - \irk(\Sigma_A C) = \dim (A^\perp) - \irk (\Sigma_{A^\perp} C),
        \end{equation}
        and
        \begin{equation}
            \label{eqn:stabilizer-complementarity-irk}
                \dim (A) - \dim (\Sigma_A C) - \irk (\Sigma_A \rad(C)) = \dim (A^\perp) - \dim (\Sigma_{A^\perp} C) - \irk (\Sigma_{A^\perp} \rad(C)).
        \end{equation}
        In particular,
        \begin{equation}
            \label{eqn:stabilizer-complementarity}
                \irk(\Sigma_A C) - \irk(\Sigma_A \rad(C)) - \dim(\Sigma_A C) = \irk(\Sigma_{A^\perp} C) - \irk(\Sigma_{A^\perp} \rad(C)) - \dim(\Sigma_{A^\perp} C).
        \end{equation}
\end{corollary}
\begin{proof}
        From Theorem~\ref{thm:cleaning-lemma}, we have $\Pi_A \rad(C) = \Pi_A C^\perp = (\Sigma_A C)^\perp$, where the orthogonal complement is taken in $A$. A similar statement holds for $A^\perp$. By Proposition~\ref{thm:dim-irk-orthogonality} applied in the symplectic space $A$, we get
        $\dim(\Pi_A \rad(C))=\dim(A)-\irk(\Sigma_A C)$ and $\irk(\Pi_A \rad(C))=\dim(A)-\dim(\Sigma_A C)$, and similarly for $A^\perp$. Hence,~\eqref{eqn:stabilizer-complementarity-dim} follows from~\eqref{eqn:complementarity-dim}, and~\eqref{eqn:stabilizer-complementarity-irk} from~\eqref{eqn:complementarity-irk}. Finally,~\eqref{eqn:stabilizer-complementarity} follows by subtracting~\eqref{eqn:stabilizer-complementarity-dim} from~\eqref{eqn:stabilizer-complementarity-irk}.
\end{proof}

Note that a special form of the corollary above is well known in the context of stabilizer states, where (\ref{eqn:stabilizer-complementarity-irk}) follows from the fact that entanglement entropies of complementary subsystems in a pure state are equal. In particular, the values of both sides of the equation quantify the amount of bipartite entanglement in the system  \cite{fattal2004entanglement}.  

\section{Invariants for Symplectic Codes}\label{sec:invariants}

In classical anticodes~\cite{ravagnani2016generalized,byrne2023tensor} a single invariant, namely the dimension of the intersection of a code with an anticode, gives rise to the generalized weights. In the symplectic setting, an intersection $C\cap A$ carries two pieces of information: its dimension, which counts the logical qubits supported on $A$, and its isorank, which also accounts for the isotropic part of $C\cap A$. In this section, we study the two corresponding invariants, show that they are related by symplectic duality (Lemma~\ref{lemma:weight-complementarity}), and introduce the associated generalized weights and profiles. These recover the minimum distance (Corollary~\ref{prop:anticode-distance}) and lead to Singleton-type bounds (Theorem~\ref{thm:gensingletonboun}). Throughout this section, we let $C\leq V^n$ be a code with $\dim(C)=k\geq 1$ and  minimum distance $d$. Whenever a result requires $C$ to be a stabilizer code, we state this assumption explicitly.

We introduce the maps $\alpha_C,\beta_C:\mA(n)\rightarrow\mathbb{Z}$ defined by
\begin{enumerate}
    \item $\alpha_C(A)=\dim(C\cap A)$,
    \item $\beta_C(A)= \irk(C\cap A)-\irk(\rad(C)\cap A)$.
\end{enumerate}
The invariant $\alpha_C(A)$ is the maximum number of pairwise orthogonal symplectic pairs in $C\cap A$ (see Remark~\ref{rem:dimirk}), while
$\beta_C(A)$ measures the isorank of this
intersection after accounting for the contribution
of $\rad(C)\cap A$. Note that $\alpha_C(A)\ge 0$ by definition, and
$\beta_C(A)\ge 0$ follows from
Theorem~\ref{thm:dim-irk-monotone},
since $\rad(C)\cap A\le C\cap A$. Hence, we actually have $\alpha_C,\beta_C:\mA(n)\rightarrow\mathbb{Z}_{\geq 0}$. 
When~$C$ is a stabilizer code, by~\eqref{eqn:stabilizer-complementarity}, we also have $\beta_C(A)-\alpha_C(A)=\beta_C(A^\perp)-\alpha_C(A^\perp)$.

\begin{lemma}\label{lemma:weight-complementarity}
    If $C$ is a stabilizer code, then $ \beta_C(A)+\alpha_C(A^\perp) = \dim (C)$.
\end{lemma}
\begin{proof}
    By~\eqref{eqn:MacWilliams-stabilizer}, we have
    \begin{equation*}
        \dim(C) = \dim(A) - \irk(\rad(C) \cap A) - \dim(A^\perp) + \dim(C \cap A^\perp) + \irk(C\cap A^\perp).
    \end{equation*}
    Recall that $\Sigma_A C \cong C \cap A$, and hence~\eqref{eqn:stabilizer-complementarity-dim} gives $\dim(A^\perp)-\irk(C\cap A^\perp)=\dim(A)-\irk(C\cap A)$. Substituting this into the identity above, we obtain
    \begin{equation*}
        \dim(C) = \irk(C \cap A) - \irk(\rad(C) \cap A) + \dim(C \cap A^\perp) = \beta_C(A)+\alpha_C(A^\perp),
    \end{equation*}
    which proves the statement.
\end{proof}

The following result establishes a relation between the maps $\alpha_C$ and $\beta_C$.

\begin{proposition}\label{prop:alpha-beta-bound}
    We have $\alpha_C(A)\leq \beta_C(A)$.
\end{proposition}
\begin{proof}
    Regarding $A\cap C$ as a code in $A$, Definition~\ref{def:dim-irk} gives
    \begin{equation*}
        \irk(A\cap C)=\dim(A\cap C)+\irk(\rad(A\cap C)),
    \end{equation*}
    since $\rad(A\cap C)$ is isotropic. Moreover, the elements of $A\cap\rad(C)$ belong to $A\cap C$ and are orthogonal to $C$, and in particular to $A\cap C$. Hence $A\cap\rad(C)\leq\rad(A\cap C)$, and Theorem~\ref{thm:dim-irk-monotone} gives $\irk(A\cap\rad(C))\leq\irk(\rad(A\cap C))$. Therefore, we obtain
    \begin{align*}
        \beta_C(A) &=\irk(A\cap C)-\irk(A\cap\rad(C))\\
        &\geq\irk(A\cap C)-\irk(\rad(A\cap C))\\
        &=\dim(A\cap C)=\alpha_C(A),
    \end{align*}
    which concludes the proof.
\end{proof}

The functions $\alpha_C$ and $\beta_C$ naturally give rise to two families of \textit{generalized weights} and \textit{profiles}. The generalized weights record the smallest dimension of an anticode needed to attain a prescribed threshold for either invariant, while the profiles record the largest value attained among anticodes of a fixed dimension. This leads to the following definitions.

\begin{definition}
    For any $a\in\{1,\ldots,k\}$, the $a$-\textbf{generalized weights} are
    \begin{enumerate}
        \item $\vartheta_a(C)=\min\{\dim(A):A\in\mA(n),\ \alpha_C(A)\geq a\}$,
        \item $\varphi_a(C)=\min\{\dim(A):A\in\mA(n),\ \beta_C(A)\geq a\}$.
    \end{enumerate}
    For any $b\in\{1,\ldots,n\}$, the $b$-\textbf{generalized profiles} are 
    \begin{enumerate}
        \item $\theta_b(C)=\max\{\alpha_C(A):A\in\mA(n),\ \dim(A)=b\}$,
        \item $\phi_b(C)=\max\{\beta_C(A):A\in\mA(n),\ \dim(A)=b\}$.
    \end{enumerate}
\end{definition}

\begin{corollary}\label{prop:theta-phi-bound}
    The following hold.
    \begin{enumerate}
        \item For any $a\in\{1,\ldots,k\}$ we have $\vartheta_a(C) \geq \varphi_a(C)$.
        \item For any $b\in\{1,\ldots,n\}$ we have $\theta_b(C) \leq \phi_b(C)$.
    \end{enumerate}
\end{corollary}
\begin{proof}
    We first prove (1). By Proposition~\ref{prop:alpha-beta-bound}, every anticode $A$ with $\alpha_C(A)\geq a$ also satisfies $\beta_C(A)\geq a$, that is
    \begin{equation*}
        \{A\in\mA(n):\alpha_C(A)\geq a\}\subseteq\{A\in\mA(n):\beta_C(A)\geq a\}.
    \end{equation*}
    Since the minimum of $\dim(A)$ over the larger set is at most the minimum over the smaller one, we obtain $\varphi_a(C)\leq\vartheta_a(C)$.
    For (2), let $A\in\mA(n)$ be an anticode with $\dim(A)=b$ and $\alpha_C(A)=\theta_b(C)$. Then Proposition~\ref{prop:alpha-beta-bound} gives
    \begin{equation*}
        \theta_b(C)=\alpha_C(A)\leq\beta_C(A)\leq\phi_b(C),
    \end{equation*}
    which concludes the proof.
\end{proof}

The following result shows that the generalized weights recover the minimum distance.

\begin{corollary}\label{prop:anticode-distance}
   We have  $d=\varphi_1(C)$.
\end{corollary}
\begin{proof}
    Let $A \in \mathcal{A}(n)$ with $\beta_C(A) \geq 1$. Then there exists $v \in A$ such that $v \in C \setminus \rad(C)$, as otherwise $C\cap A=\rad(C)\cap A$ and $\beta_C(A)=0$. Hence, by the definition of minimum distance, we have  $d \leq \wt(v) \leq \dim(A)$, and since $A$ was arbitrary, $d$ is at most the minimum over all such $A$. 

    Conversely, let $v\in C\setminus\rad(C)$ with $\wt(v)=d$ and set $A=A_{\operatorname{supp}(v)}$, so that $\dim(A)=d$ and $v\in A\cap C$. We claim that $\beta_C(A)\ge 1$, that is
    \begin{equation*}
        \irk(A\cap C)>\irk(A\cap\rad(C)).
    \end{equation*}
    As above, $A\cap\rad(C)$ is isotropic and contained in $\rad(A\cap C)$. 
    Suppose, by contradiction, that equality holds. Since
    $\irk(A\cap C)=\dim(A\cap C)+\dim_\F(\rad(A\cap C))$ and 
    \begin{equation*}
    \irk(A\cap\rad(C))=\dim_\F(A\cap\rad(C))\leq\dim_\F(\rad(A\cap C)),
    \end{equation*}
    it follows that $\dim(A\cap C)=0$ and $A\cap\rad(C)=\rad(A\cap C)$. In particular, the space $A\cap C$ is isotropic, so that $A\cap C=\rad(A\cap C)=A\cap\rad(C)$. Hence $v\in\rad(C)$, which contradicts the choice of~$v$.
\end{proof}

\begin{remark}\label{rem:dim-less-d}
    Observe that, by the definition of minimum distance, we have $A\cap C=\rad(C)\cap A$ whenever $\dim(A)<d$, since every element of $A\cap C$ has weight at most $\dim(A)$. In particular, $A\cap C$ is isotropic, and therefore $\phi_b(C)=\theta_b(C)=0$ for all $b<d$.
\end{remark}

\begin{example}\label{exa:repetion-inv}
Let $C \leq V^2$ be the code defined in Example~\ref{example:repetition_code3}, and recall that we have~$C = \mathrm{span}_{\mathbb{F}_2}\{(e,e), (f,f), (f,0)\}$ and  $C^{\perp} = \mathrm{span}_{\mathbb{F}_2}\{(f,f)\}$. We have already observed that $\dim(C) = 1$, $\irk(C) = 2$, and $\irk(\rad(C)) = 1$.  
One can verify that the set of anticodes $\mathcal{A}(2)$ contains the zero subspace, $V^2$ and the spaces below.
\[
A_1 = \mathrm{span}_{\mathbb{F}_2}\{(e,0),(f,0)\}, \qquad
A_2 = \mathrm{span}_{\mathbb{F}_2}\{(0,e),(0,f)\}.
\]
The following hold.
\begin{itemize}
    \item $\theta_1(C) = 0$ and $\theta_2(C) = 1$. We have $C\cap A_1=\mathrm{span}_{\mathbb{F}_2}\{(f,0)\}$ and $C\cap A_2=\mathrm{span}_{\mathbb{F}_2}\{(0,f)\}$, which are isotropic, and hence $\dim(C\cap A_1)=\dim(C\cap A_2)=0$. Moreover, $\dim(C\cap V^2)=\dim(C)=1$.
    
    \item $\phi_1(C) = \phi_2(C) = 1$. Since $\rad(C)=\mathrm{span}_{\mathbb{F}_2}\{(f,f)\}$ intersects $A_1$ and $A_2$ trivially, we have $\beta_C(A_1)=\irk(C\cap A_1)=1$ and $\beta_C(A_2)=\irk(C\cap A_2)=1$. Moreover, $\beta_C(V^2)=\irk(C)-\irk(\rad(C))=2-1=1$.
    
    \item $\varphi_1(C)=1$. This can be attained for the anticode $A_1$. In particular, we have $C\cap A_1=\mathrm{span}_{\mathbb{F}_2}\{(f,0)\}$ and $\rad(C)\cap A_1=\mathrm{span}_{\mathbb{F}_2}\{(0,0)\}$, which implies
    \begin{equation*}
        \irk(C\cap A_1)-\irk(\rad(C)\cap A_1)=\irk(C\cap A_1)=1.
    \end{equation*}
    Moreover, we have $\varphi_1(C)=1=d$, in line with Corollary~\ref{prop:anticode-distance}.
\end{itemize}
\end{example}

\begin{remark}\label{rem:galconn}
    Observe that, by Corollary~\ref{prop:theta-phi-bound}, we have that $\theta_b(C)\leq \phi_b(C)$ for any~$b\in\{1,\ldots,n\}$. Since $\alpha_C$ is monotone by Theorem~\ref{thm:dim-irk-monotone}, and $\beta_C$ is monotone by Lemma~\ref{lemma:weight-complementarity} when $C$ is a stabilizer code, as a consequence of~\cite[Theorem~3.2]{xu2022galois}, we have that $(\vartheta_a, \theta_b)$ forms a Galois connection between the sets $\{1, \ldots, k\}$ and $\{1, \ldots, n\}$, and so does $(\varphi_a, \phi_b)$ when $C$ is a stabilizer code. In particular, we have
    \begin{enumerate}
        \item $a\leq \theta_b(C)$ if and only if $\vartheta_a(C)\leq b$,
        \item if $C$ is a stabilizer code, then $a\leq \phi_b(C)$ if and only if $\varphi_a(C)\leq b$.
    \end{enumerate}
\end{remark}

We now establish some relations between consecutive generalized profiles.

\begin{proposition}\label{prop:relinv}
Let $b\in\{1,\ldots,n-1\}$, and let $A,A'\in\mA(n)$ with $A'\leq A$, $\dim(A')=b$, and $\dim(A)=b+1$. Then
\begin{equation}\label{eq:reinv}
    \alpha_C(A) + \beta_C(A) \leq \alpha_C(A')+\beta_C(A')+2.
\end{equation}
Moreover, we have $\theta_{b+1}(C) \leq \theta_b(C) + 2$ and $\phi_{b+1}(C) \leq \phi_b(C) + 2$.
\end{proposition}

\begin{proof}
    Let $A'\leq A$ be as in the statement. Since $A'$ is obtained from $A$ by removing one coordinate, $C\cap A'$ is the kernel of the projection of $C\cap A$ onto this coordinate, and hence $\dim_\F(C\cap A)-\dim_\F(C\cap A')\leq 2$. Moreover, since $\rad(C)\cap A$ is isotropic, equation~\eqref{eqn:dim-irk-decomp} gives
    \begin{equation*}
        \alpha_C(A)+\beta_C(A)=\dim_\F(C\cap A)-\dim_\F(\rad(C)\cap A),
    \end{equation*}
    and similarly for $A'$. Since $\rad(C)\cap A'\leq\rad(C)\cap A$, we obtain
    \begin{align*}
        \alpha_C(A)+\beta_C(A)
        &=\dim_\F(C\cap A)-\dim_\F(\rad(C)\cap A)\\
        &\leq\dim_\F(C\cap A')+2-\dim_\F(\rad(C)\cap A')\\
        &=\alpha_C(A')+\beta_C(A')+2,
    \end{align*}
    which proves~\eqref{eq:reinv}.
    Now let $A\in\mA(n)$ be such that $\dim(A)=b+1$ and $\alpha_C(A)=\theta_{b+1}(C)$, and let $A'\leq A$ with $\dim(A')=b$. By equation~\eqref{eqn:dim-irk-decomp}, the difference $\dim_\F(C\cap A)-\dim_\F(C\cap A')\leq 2$ is the sum of the differences of the dimensions and of the isoranks of $C\cap A$ and $C\cap A'$, both of which are nonnegative by Theorem~\ref{thm:dim-irk-monotone}. Hence
    \begin{equation*}
        \theta_{b+1}(C)=\alpha_C(A)\leq\alpha_C(A')+2\leq\theta_b(C)+2.
    \end{equation*}
    Similarly, let $A\in\mA(n)$ be such that $\dim(A)=b+1$ and $\beta_C(A)=\phi_{b+1}(C)$, and let $A'\leq A$ with $\dim(A')=b$. The same argument gives $\irk(C\cap A)\leq\irk(C\cap A')+2$. Moreover, since $\rad(C)\cap A'\leq\rad(C)\cap A$, Theorem~\ref{thm:dim-irk-monotone} gives $\irk(\rad(C)\cap A')\leq\irk(\rad(C)\cap A)$. As a consequence, we obtain
    \begin{align*}
        \phi_{b+1}(C)
        &=\irk(C\cap A)-\irk(\rad(C)\cap A)\\
        &\leq\irk(C\cap A')+2-\irk(\rad(C)\cap A')\\
        &=\beta_C(A')+2\\
        &\leq\phi_b(C)+2,
    \end{align*}
    which concludes the proof.
\end{proof}

The following example illustrates the properties of the generalized profiles established in the proposition above.

\begin{example}
    Let $C\leq V^4$ be the code as in Example~\ref{example:2x2_bacon_shor_new}, that is 
\[
C = \mathrm{span}_{\mathbb{F}_2}\{(e,e,0,0), (f,f,0,0), (0,0,e,e), (0,0,f,f), (e,0,e,0), (f,0,f,0)\},
\]
and observe that $C$ admits the orthogonal decomposition $C=\rad(C)\oplus K$, with 
\begin{equation*}
    K=\mathrm{span}_{\mathbb{F}_2}\{(e,e,0,0), (f,f,0,0), (e,0,e,0), (f,0,f,0)\},\qquad\textup{and}\qquad \rad(C)=C^\perp.
\end{equation*}
Let $\mA(4)$ denote the family of anticodes in $V^4$. One can check that the following hold.
\begin{enumerate}
    \item $\theta_1(C)=\phi_1(C)=0$. This follows from the fact that the codewords of $C$ have weight at least $2$. Hence, no $1$-dimensional anticode intersects $C$ nontrivially.
    
    \item $\theta_2(C)=0$. This follows from the fact that $C\cap A=\rad(C\cap A)$ for any $A\in\mA(4)$ with $\dim(A)=2$.
    
    \item $\theta_3(C)=2$. This value is attained, for example, by intersecting $C$ with the $3$-dimensional anticode $\mathrm{span}_{\mathbb{F}_2}\{(e,0,0,0),(f,0,0,0),(0,e,0,0),(0,f,0,0),(0,0,e,0),(0,0,f,0)\}$, which leads to the $2$-dimensional space $\mathrm{span}_{\mathbb{F}_2}\{(e,e,0,0),(f,f,0,0),(e,0,e,0),(f,0,f,0)\}$. One can also observe that, for any $A\in\mA(4)$ with $\dim(A)=3$, we have $\dim(C\cap A)=2$ and $\rad(C\cap A)=\{0\}$.
    
    \item $\theta_4(C)=2$. This follows from the fact that the only anticode of dimension $4$ is $V^4$. In particular, $\dim(C\cap V^4)=\dim(C)=\frac{1}{2}\dim_\F(K)$.
    
    \item $\phi_2(C)=\phi_3(C)=2$. Note that, for any $A\in\mA(4)$ with $\dim(A)\in\{2,3\}$, we have $\irk(C\cap A)=2$. Moreover, $\rad(C)\cap A=\{0\}$, since every nonzero element of $\rad(C)$ has support $\{1,2,3,4\}$. 
    
    \item $\phi_4(C)=2$. This is a consequence of the fact that $\irk(C\cap V^4)=\irk(C)=4$ and~$\irk(\rad(C)\cap V^4)=\irk(\rad(C))=2$.
\end{enumerate}
\end{example}

\begin{proposition}\label{prop:vartheta-phi}
    For any $a\in\{1,\ldots,k-2\}$, we have
    \begin{enumerate}
        \item $\vartheta_{a}(C)+ 1 \leq \vartheta_{a+2}(C)$,
        \item $\varphi_{a}(C) + 1 \leq \varphi_{a+2}(C)$.
    \end{enumerate}
\end{proposition}
\begin{proof}
    Let $a\in\{1,\ldots,k-2\}$, and let $A\in\mA(n)$ be an anticode with $\dim(A)=\vartheta_{a+2}(C)$ and $\alpha_C(A)\geq a+2$. Since $A\neq\{0\}$, we can choose an anticode $A'\leq A$ with $\dim(A')=\dim(A)-1$. As in the proof of Proposition~\ref{prop:relinv}, we have $\alpha_C(A)\leq\alpha_C(A')+2$, and hence $\alpha_C(A')\geq a$. Therefore,
    \begin{equation*}
        \vartheta_a(C)\leq\dim(A')=\vartheta_{a+2}(C)-1.
    \end{equation*}
    The same argument, using $\beta_C(A)\leq\beta_C(A')+2$, which is also shown in the proof of Proposition~\ref{prop:relinv}, gives $\varphi_a(C)+1\leq\varphi_{a+2}(C)$.
\end{proof}

The bound~\eqref{eq:reinv} of Proposition~\ref{prop:relinv} shows that the sum $\alpha_C+\beta_C$ increases by at most two when a coordinate is added to an anticode. In particular, an increase of two in one of them precludes an increase in the other. This observation motivates the following definition.

\begin{definition}
    For any $a\in\{1,\ldots,k\}$, we let 
\begin{equation}\label{eqn:delta-definition}
    \delta_a(C) = \min\{\dim(A) : A\in\mA(n),\ \alpha_C(A)+\beta_C(A) \geq 2a\}.
\end{equation}
\end{definition}

\begin{proposition}\label{prop:GQWmono}
    For any $a\in\{1,\ldots,k-1\}$, we have $\delta_a(C) +1 \leq \delta_{a+1}(C)$.
\end{proposition}
\begin{proof}
    Let $A \in \mA(n)$ be an anticode such that $\dim(A)=\delta_{a+1}(C)$ and $\alpha_C(A)+\beta_C(A)\ge 2(a+1)$. Let $A' \in \mA(n)$ be any anticode with $A' \le A$ and $\dim(A')=\dim(A)-1=\delta_{a+1}(C)-1$. By~\eqref{eq:reinv}, we have 
\begin{equation*}
\alpha_C(A') + \beta_C(A') \ge
\alpha_C(A) + \beta_C(A) - 2 \ge 2(a+1)-2 =2a.
\end{equation*}
Finally, we have that  $\delta_a(C)\le \dim(A')=\delta_{a+1}(C)-1$.\qedhere
\end{proof}

We obtain the following bounds.

\begin{theorem}\label{thm:gensingletonboun}If $C$ is a stabilizer code, then $\varphi_a(C)\leq n-d-k+a+1$ for every $a\in\{1,\ldots,k\}$.
\end{theorem}
\begin{proof}
We first observe that $\alpha_C(A) = \beta_C(A) = 0$ for any $A\in\mA(n)$ with $\dim(A) < d$. Indeed, if $v\in C\cap A$ then $\wt(v)\le\dim(A) < d$, which by definition of minimum distance forces $v\in\rad(C)$. Therefore $C\cap A = \rad(C)\cap A$ is isotropic, and so $\alpha_C(A) = \dim(C\cap A) = 0$ and $\beta_C(A) = \irk(C\cap A) - \irk(\rad(C)\cap A) = 0$.

Now let $A\in\mA(n)$ be an anticode with $\dim(A^{\perp}) = d-1$, so that $\dim(A) = n-d+1$. By the above we have $\alpha_C(A^{\perp}) = \beta_C(A^{\perp}) = 0$, and Lemma~\ref{lemma:weight-complementarity}, applied to $A$ and to $A^{\perp}$, gives
\begin{equation*}
    \beta_C(A) = k - \alpha_C(A^{\perp}) = k
    \qquad\text{and}\qquad
    \alpha_C(A) = k - \beta_C(A^{\perp}) = k.
\end{equation*}
In particular, $\alpha_C(A) + \beta_C(A) = 2k$, and hence $\delta_k(C)\le\dim(A) = n-d+1$ by~\eqref{eqn:delta-definition}.
Since~$\delta_a(C)$ is strictly increasing in $a$, by Proposition~\ref{prop:GQWmono}, we have
\begin{equation*}
    \delta_a(C)\le\delta_{a+1}(C) - 1\le\cdots\le\delta_k(C) - (k-a)
    \le n-d-k+a+1.
\end{equation*}
Now, by Proposition~\ref{prop:alpha-beta-bound} we have $\alpha_C(A) + \beta_C(A) \leq 2\beta_C(A)$ and hence
$$\{A \in \mathcal{A}(n) \::\: \alpha_C(A) + \beta_C(A) \geq 2a\} \subseteq \{A \in \mathcal{A}(n) \::\: \beta_C(A) \geq a\}.$$
Hence, taking minima of the dimensions over these sets gives $\varphi_a(C) \leq \delta_a(C)$, which concludes the proof.
\end{proof}

The result above can be seen as a generalization of the quantum Singleton bound~\cite{knill1997theory,rains1999nonbinary,preskill1998quantum}, which is recovered by taking $a = 1$.

\begin{theorem}[Quantum Singleton Bound]\label{thm:qsingleton}
If $C$ is a stabilizer code, then $2(d-1)\leq n-k$.
\end{theorem}
\begin{proof}
By Corollary~\ref{prop:anticode-distance}, we have $d=\varphi_1(C)$. Hence Theorem~\ref{thm:gensingletonboun}, applied with $a=1$, gives
    \begin{equation*}
        d=\varphi_1(C)\leq n-d-k+2,
    \end{equation*}
    which is equivalent to $2(d-1)\leq n-k$.
\end{proof}

The following can be seen as a \textit{lower Singleton-type bound}.

\begin{proposition}
    For any $a\in\{1,\ldots,k\}$, we have $ \varphi_a(C)\geq a$.
\end{proposition}
\begin{proof}
Let $A\in\mA(n)$ with $\beta_C(A)\geq a$. By Theorem~\ref{thm:dim-irk-monotone}, we have $$a\leq\beta_C(A)\leq\irk(C\cap A)\leq\irk(A)=\dim(A).$$ Since $A$ was arbitrary, we obtain $\varphi_a(C)\geq a$.
\end{proof}

\section{Binomial Moments and Enumerators}

In this section, we introduce and study further invariants defined through anticodes, namely the \textit{binomial moments} and the \textit{weight distribution} of a code. These notions are direct analogues of classical invariants that have been extensively studied for several metrics (see, for example,~\cite{ravagnani2016generalized,byrne2020rank,byrne2023tensor,cotardo2025zeta}).  We also introduce the corresponding weight enumerators and show how they can be expressed in terms of the binomial moments and the weight distribution. These invariants naturally fit within the framework of \cite{rains2002quantum}; the binomial moments we define below are related to enumerators $A'$ and $B'$ of that work. Throughout this section, we assume that $q$ is a prime power and that $\dim_{\F_q}(V)=2$. In particular, $|V|=q^2$. Moreover, we let $C\leq V^n$ be a code with $\dim(C)=k$ and $\irk(C)=s$.  The invariants introduced in this section depend only on the $\mathbb{F}$-linear structure of the code, through the intersections with anticodes, and therefore fall within the scope of the theory of demimatroids and their Wei-type dualities~\cite{britz2012wei,britz2015demi}. Compared with~\cite{shor1997quantum,rains2002quantum}, the anticode description expresses the weight enumerators entirely in terms of the vector space dimensions of the intersections $C\cap A_J$ and $C^\perp\cap A_J$, that is, in terms of data attached to individual subsystems. In this form, the MacWilliams-type identities of Theorem~\ref{thm:macwilliams-moments} follow directly from the duality $A_J\mapsto A_J^\perp=A_{J^c}$ and Theorem~\ref{thm:duality-distance}, without the use of characters. Moreover, this is the kind of local information that arises when a code is assembled from smaller pieces, as in the tensor network constructions of~\cite{cao2023quantum,cao2024quantum}. The section is organized in two parts. In Section~\ref{subsec:binomial}, we establish the correspondence between binomial moments and weight distributions, together with the corresponding MacWilliams-type identities. In Section~\ref{subsec:greene}, we show that the weight enumerators of a code are determined by a single function on anticodes, namely the vector space dimension of the puncturings of the code, and we prove a quantum analogue of Greene's theorem~\cite{greene1976weight} (Theorem~\ref{thm:greene}).

\subsection{Binomial moments and weight distributions}\label{subsec:binomial}

In this subsection, we introduce the notions of weight distribution and binomial moments of a code, we show that they determine each other, and we derive the corresponding MacWilliams identities for binomial moments.

\begin{definition}~\label{def:binomial-moments}
    The \textbf{weight distribution} of $C$ is the tuple of length $n+1$ whose $a$-th component is
    \begin{equation*}
        \mathcal{W}_a(C)=\sum_{\substack{A\in\mathcal{A}(n)\\\dim(A)=a}}\mathcal{W}_A(C),\quad\textup{ where }\quad \mathcal{W}_A(C)=|\{c\in C\mid\textup{supp}(c)=\textup{supp}(A)\}|.
    \end{equation*}
    For any $b\in\{0,\ldots,n\}$, the $b$\textbf{-th binomial moment} of $C$ is
    \begin{equation*}
        \mathcal{B}_b(C)=\sum_{\substack{A\in\mathcal{A}(n)\\\dim(A)=b}}\mathcal{B}_A(C),\quad\textup{ where }\quad \mathcal{B}_A(C)=|C\cap A|.
    \end{equation*}
\end{definition}

Clearly, we have $\mathcal{B}_A(C) = q^{\dim_\F(C \cap A)}$. The following result establishes a correspondence between weight distribution and binomial moments of a code. In particular, it shows that these two invariants are equivalent in the sense that they encode the same information. This result can be viewed as the quantum analogue of~\cite[Theorem~6.4]{byrne2023tensor} (see also~\cite[Lemma~30]{ravagnani2016rank} and~\cite[Theorem~3.8]{byrne2020rank}). 

\begin{lemma}\label{lem:relations}
    The following hold for any $A\in\mathcal{A}(n)$ and $a,b\in\{0,\ldots,n\}$.
    \begin{enumerate}
        \item $\displaystyle \mathcal{B}_A(C)=\sum_{\substack{A'\in\mathcal{A}(n)\\A'\leq A}}\mathcal{W}_{A'}(C)$.
        \item $\displaystyle \mathcal{W}_A(C)=\sum_{\substack{A'\in\mathcal{A}(n)\\A'\leq A}}(-1)^{\dim(A)-\dim(A')}\mathcal{B}_{A'}(C)$.
        \item $\displaystyle \mathcal{B}_b(C)=\sum_{a=0}^b\binom{n-a}{b-a}\mathcal{W}_{a}(C)$.
        \item $\displaystyle \mathcal{W}_a(C)=\sum_{b=0}^a(-1)^{a-b}\binom{n-b}{a-b}\mathcal{B}_{b}(C)$.
    \end{enumerate}
\end{lemma}
\begin{proof}
    The first two equalities follow directly from the definitions of weight distribution and binomial moments, together with the fact that the poset $\mathcal{A}(n)$ of anticodes ordered by inclusion forms a lattice isomorphic to the subset lattice of $\{1, \ldots, n\}$. Now, by (1) we have
    \begin{align*}
        \mathcal{B}_b(C)&=\sum_{\substack{A\in\mathcal{A}(n)\\\dim(A)=b}}\sum_{\substack{A'\in\mathcal{A}(n)\\A'\leq A}}\mathcal{W}_{A'}(C)\\        &=\sum_{a=0}^b\sum_{\substack{A'\in\mathcal{A}(n)\\\dim(A')=a}}\mathcal{W}_{A'}(C)|\{A\in\mathcal{A}(n):\dim(A)=b, A'\leq A\}|\\
        &=\sum_{a=0}^b\binom{n-a}{b-a}\mathcal{W}_{a}(C),
    \end{align*}
   since there are exactly $\binom{n - a}{b - a}$ anticodes of dimension $b$ that contain a fixed anticode $A'$ of dimension~$a$. Finally, a similar argument applied to (1) and~(2) establishes (4).
\end{proof}

The duality relation established in Theorem~\ref{thm:duality-distance} leads to the following MacWilliams identities for binomial moments. One can observe that the next result is the quantum analogue of ~\cite[Theorem~6.7]{byrne2023tensor} (see also~\cite[Theorem~7.1]{byrne2020rank} and~\cite[Lemma~28]{ravagnani2016rank}).

\begin{theorem}\label{thm:macwilliams-moments} 
    The following hold.
    \begin{enumerate}
        \item $\displaystyle \mathcal{B}_A(C^\perp) = q^{2\dim(A) - k-s} \mathcal{B}_{A^\perp}(C)$ for any $A\in\mathcal{A}(n)$.
        \item $\displaystyle \mathcal{B}_b(C^\perp) = q^{2b-k-s} \mathcal{B}_{n-b}(C)$ for any $b\in\{0,\ldots,n\}$.
    \end{enumerate}
\end{theorem}
\begin{proof}
    The first equation follows immediately from~\eqref{eqn:MacWilliams}. For the second, we observe that the map $A \mapsto A^\perp$ naturally induces a bijection between anticodes of dimension $b$ and those of dimension $n - b$.
\end{proof}

We introduce the following notions of homogeneous polynomials that encode structural properties of a code.

\begin{definition}
    Let $x$ and $y$ be indeterminates. The \textbf{weight enumerators} of $C$ are 
    \begin{equation*}
        \mathcal{A}(x,y;C) = \sum_{c\in \rad(C)} x^{\wt(c)}y^{n-\wt(c)}\qquad\textup{ and }\qquad \mathcal{B}(x,y;C) = \sum_{c\in C} x^{\wt(c)}y^{n-\wt(c)}.
    \end{equation*}
\end{definition}

One can observe that $\mathcal{B}(x,y;C)$ corresponds to the classical weight enumerator of a code. However, as for the minimum distance, the elements of $\rad(C)$ require separate consideration. Therefore, we introduce a second enumerator, $\mathcal{A}(x,y;C)$, which counts specifically the weights of the elements of $\rad(C)$. This approach aligns with the conventional definition of quantum weight enumerators as in \cite{shor1997quantum}. Note that our normalization (or lack thereof) conforms to convention in \cite{shor1997quantum} as opposed to that of \cite{rains2002quantum}. We omit the proof of the next result, as it follows directly from the definitions of weight enumerators and minimum distance.

\begin{proposition}
    The minimum distance of $C$ is the trailing degree of~$\mathcal{B}(x,1;C)-\mathcal{A}(x,1;C)$.
\end{proposition}

As a consequence of Lemma~\ref{lem:relations}, we can express the weight enumerator
$\mathcal{B}(x,y;C)$ in terms of the weight distribution and the binomial moments of the code, as
shown in the next result.

\begin{proposition}\label{prop:binomial-to-enumerator}
    The following holds.
    \begin{equation*}
        \mathcal{B}(x, y; C)=\sum_{a=0}^n\mathcal{W}_a(C)x^ay^{n-a}=\sum_{b=0}^n\mathcal{B}_b(C)x^b(y-x)^{n-b}.
    \end{equation*}
\end{proposition}
\begin{proof}
    By definition of weight enumerator, we get
    \begin{align*}
         \mathcal{B}(x, y; C)=\sum_{a=0}^n\sum_{\substack{c\in C\\\wt(c)=a}}x^ay^{n-a}=\sum_{a=0}^n|\{c\in C\mid \wt(c)=a\}|x^ay^{n-a}=\sum_{a=0}^n\mathcal{W}_a(C)x^ay^{n-a}.
    \end{align*}
    Combining this with Lemma~\ref{lem:relations}, we get
    \begin{align*}
        \mathcal{B}(x, y; C)&=\sum_{a=0}^n\sum_{b=0}^a(-1)^{a-b}\binom{n-b}{a-b}\mathcal{B}_{b}(C)x^ay^{n-a}\\
        &=\sum_{b=0}^n\left(\sum_{a=b}^n(-1)^{a-b}\binom{n-b}{a-b}x^ay^{n-a}\right)\mathcal{B}_{b}(C)\\
        &=\sum_{b=0}^n\left(\sum_{a=0}^{n-b}(-1)^{a}\binom{n-b}{a}x^{a+b}y^{n-a-b}\right)\mathcal{B}_{b}(C)\\
        &=\sum_{b=0}^nx^b\left(\sum_{a=0}^{n-b}(-1)^{a}\binom{n-b}{a}x^{a}y^{n-b-a}\right)\mathcal{B}_{b}(C)\\      
        &=\sum_{b=0}^n\mathcal{B}_{b}(C)x^b(y-x)^{n-b},
    \end{align*}
    which concludes the proof.
\end{proof}

We conclude this section with the following example.

\begin{example}\label{example:repetition_code5}
Let $C$ and $C^\perp$ be as in Example~\ref{exa:repetion-inv}, that is $C = \mathrm{span}_{\mathbb{F}_2}\{(e,e), (f,f), (f,0)\}$ and 
$C^{\perp} = \mathrm{span}_{\mathbb{F}_2}\{(f,f)\}$. Listing all codewords, we get
\begin{align*}
    C
    &= \{(0,0), (e,e), (f,f), (e+f,e+f), (f,0), (e+f,e), (0,f), (e,e+f)\},\\[1ex]
    \rad(C)
    &= C^\perp
     = \{(0,0), (f,f)\}.
\end{align*}
From this description, the weight enumerators follow immediately.  
The code $C$ has one codeword of weight $0$, two of weight $1$, and five of weight $2$, hence $\mathcal{B}(x,y;C) = y^2 + 2xy + 5x^2$. Similarly, $C^\perp$ has one codeword of weight $0$ and one of weight $2$, and therefore we have $\mathcal{B}(x,y;C^\perp) = \mathcal{A}(x,y;C) = y^2 + x^2$, since $C^\perp = \rad(C)$. On the other hand, from Example~\ref{exa:repetion-inv} we know that $\dim_{\mathbb{F}_2}(C \cap A) = 1$ and $C^\perp \cap A = \{(0,0)\}$ for each anticode $A \in \{A_{\{1\}}, A_{\{2\}}\}$. Therefore,
\begin{equation*}
    \mathcal{B}_1(C) = 4,
    \qquad
    \mathcal{B}_2(C) = |C\cap V^2|= |C| = 8,\qquad 
    \mathcal{B}_1(C^\perp) = 2,
    \qquad
    \mathcal{B}_2(C^\perp) =|C^\perp\cap V^2|= 2.
\end{equation*}
Simple algebraic computations lead to
\begin{align*}
    \mathcal{B}(x,y;C)
    &= (y-x)^2 + 4x(y-x) + 8x^2
    = y^2 + 2xy + 5x^2,\\
    \mathcal{A}(x,y;C)&=\mathcal{B}(x,y;C^\perp)= (y-x)^2 + 2x(y-x) + 2x^2= y^2 + x^2,
\end{align*}
as expected, in line with Proposition~\ref{prop:binomial-to-enumerator}.
\end{example}

\subsection{A quantum Greene theorem}\label{subsec:greene}
In this subsection, we show that the weight enumerators of a code are determined by a single function on the set of anticodes, namely the vector space dimension of the puncturings of the code. This provides an analogue of Greene's theorem~\cite{greene1976weight}, which expresses the weight enumerator of a classical linear code in terms of the Tutte polynomial of its matroid.

\begin{definition}\label{def:rho}
    For any $A\in\mA(n)$, we let $\rho_C(A)=\dim_{\F_q}(\Pi_A C)$. The \emph{rank generating polynomial} of $C$ is
    \begin{equation*}
        R_C(u,v)=\sum_{A\in\mA(n)}u^{\dim_{\F_q}(C)-\rho_C(A)}\,v^{2\dim(A)-\rho_C(A)}.
    \end{equation*}
\end{definition}

Note that $\rho_C(V^n)=\dim_{\F_q}(C)$ and $\rho_C(A)\leq\dim_{\F_q}(A)=2\dim(A)$ for any $A\in\mA(n)$, so that $R_C(u,v)$ is a polynomial in the indeterminates $u$ and $v$. Moreover, if $A$ is supported on $J$, then the kernel of the restriction of $\pi_J$ to $C$ is $C\cap A^\perp$ by Proposition~\ref{prop:ant-compl}, and hence
\begin{equation}\label{eqn:rho-kernel}
    \rho_C(A)=\dim_{\F_q}(C)-\dim_{\F_q}(C\cap A^\perp).
\end{equation}

\begin{remark}\label{rem:polymatroid}
    The function $\rho_C$ is nondecreasing and satisfies
    \begin{equation*}
        \rho_C(A_1+A_2)+\rho_C(A_1\cap A_2)\leq\rho_C(A_1)+\rho_C(A_2)
    \end{equation*}
    for any $A_1,A_2\in\mA(n)$. Indeed, by Proposition~\ref{prop:ant-compl}, we have $(A_1+A_2)^\perp=A_1^\perp\cap A_2^\perp$ and $(A_1\cap A_2)^\perp=A_1^\perp+A_2^\perp$, so that 
    \[C\cap(A_1+A_2)^\perp=(C\cap A_1^\perp)\cap(C\cap A_2^\perp) \text{ and } (C\cap A_1^\perp)+(C\cap A_2^\perp)\leq C\cap(A_1\cap A_2)^\perp.
    \]
    The inequality then follows from~\eqref{eqn:rho-kernel} and~\eqref{eqn:dimF-modularity}. Therefore, identifying each anticode with its support, $\rho_C$ is the rank function of a polymatroid on $\{1,\ldots,n\}$, and $R_C(u,v)$ is the analogue of its Whitney rank generating function; see~\cite{greene1976weight} for the case of classical linear codes and~\cite{vertigan2004latroids} for codes over modules.
\end{remark}

The following result is the analogue, for the rank generating polynomial, of the duality between the Tutte polynomials of a matroid and of its dual.

\begin{proposition}\label{prop:poly-dual}
    For any $A\in\mA(n)$, we have
    \begin{equation*}
        \rho_{C^\perp}(A)=2\dim(A)-\dim_{\F_q}(C)+\rho_C(A^\perp).
    \end{equation*}
    In particular, $R_{C^\perp}(u,v)=R_C(v,u)$.
\end{proposition}
\begin{proof}
    By Theorem~\ref{thm:cleaning-lemma}, we have $\Pi_A C^\perp=(\Sigma_A C)^\perp$, where the orthogonal is taken in $A$. Since $\Sigma_A C\cong C\cap A$ and $\dim_{\F_q}(A)=2\dim(A)$, we get $\rho_{C^\perp}(A)=2\dim(A)-\dim_{\F_q}(C\cap A)$. The first statement now follows from~\eqref{eqn:rho-kernel} applied to $A^\perp$. In particular, we have $\dim_{\F_q}(C^\perp)-\rho_{C^\perp}(A)=2\dim(A^\perp)-\rho_C(A^\perp)$ and $2\dim(A)-\rho_{C^\perp}(A)=\dim_{\F_q}(C)-\rho_C(A^\perp)$. The second statement follows by substituting these expressions into the definition of $R_{C^\perp}(u,v)$, since the map $A\mapsto A^\perp$ is a bijection of $\mA(n)$.
\end{proof}

We are now ready to establish the main result of this subsection.

\begin{theorem}[Quantum Greene Theorem]\label{thm:greene}
    Let $\rho=\dim_{\F_q}(C)=\rho_C(V^n)$. We have
    \begin{equation*}
        \mathcal{B}(x,y;C)=x^{n-\rho/2}(y-x)^{\rho/2}\,R_C\!\left(\sqrt{\frac{q^2x}{y-x}},\ \sqrt{\frac{y-x}{x}}\right).
    \end{equation*}
\end{theorem}
\begin{proof}
    By Proposition~\ref{prop:binomial-to-enumerator} and Definition~\ref{def:binomial-moments}, we have
    \begin{equation*}
        \mathcal{B}(x,y;C)=\sum_{A\in\mA(n)}q^{\dim_{\F_q}(C\cap A)}\,x^{\dim(A)}(y-x)^{n-\dim(A)}.
    \end{equation*}
    By~\eqref{eqn:rho-kernel} applied to $A^\perp$, we have $\dim_{\F_q}(C\cap A)=\rho-\rho_C(A^\perp)$. Since the map $A\mapsto A^\perp$ is a bijection of $\mA(n)$ and $\dim(A^\perp)=n-\dim(A)$, we obtain
    \begin{equation}\label{eqn:greene-proof}
        \mathcal{B}(x,y;C)=\sum_{A\in\mA(n)}q^{\rho-\rho_C(A)}\,x^{n-\dim(A)}(y-x)^{\dim(A)}.
    \end{equation}
    On the other hand, the term of the right-hand side of the statement indexed by $A$ is
    \begin{equation*}
        x^{n-\rho/2}(y-x)^{\rho/2}\left(\frac{q^2x}{y-x}\right)^{\frac{\rho-\rho_C(A)}{2}}\left(\frac{y-x}{x}\right)^{\frac{2\dim(A)-\rho_C(A)}{2}}=q^{\rho-\rho_C(A)}\,x^{n-\dim(A)}(y-x)^{\dim(A)},
    \end{equation*}
    which concludes the proof.
\end{proof}

As a consequence, we recover the MacWilliams identity for weight enumerators of symplectic codes.

\begin{corollary}\label{cor:macwilliams-enum}
    We have
    \begin{equation*}
        \mathcal{B}(x,y;C^\perp)=\frac{1}{|C|}\,\mathcal{B}\big(y-x,\,y+(q^2-1)x;\,C\big).
    \end{equation*}
    In particular, if $C$ is a stabilizer code then $\mathcal{A}(x,y;C)=\frac{1}{|C|}\,\mathcal{B}\big(y-x,\,y+(q^2-1)x;\,C\big)$.
\end{corollary}
\begin{proof}
    Let $\rho=\dim_{\F_q}(C)$, so that $\dim_{\F_q}(C^\perp)=2n-\rho$. By~\eqref{eqn:greene-proof} applied to $C^\perp$ and Proposition~\ref{prop:poly-dual}, we have
    \begin{align*}
        \mathcal{B}(x,y;C^\perp)&=\sum_{A\in\mA(n)}q^{2\dim(A^\perp)-\rho_C(A^\perp)}\,x^{n-\dim(A)}(y-x)^{\dim(A)}\\
        &=\sum_{A\in\mA(n)}q^{2\dim(A)-\rho_C(A)}\,x^{\dim(A)}(y-x)^{n-\dim(A)},
    \end{align*}
    where the second equality follows by replacing $A$ with $A^\perp$. On the other hand, by~\eqref{eqn:greene-proof} we have
    \begin{equation*}
        \frac{1}{|C|}\,\mathcal{B}\big(y-x,\,y+(q^2-1)x;\,C\big)=\sum_{A\in\mA(n)}q^{-\rho_C(A)}\,(y-x)^{n-\dim(A)}(q^2x)^{\dim(A)},
    \end{equation*}
    which proves the first statement. The second follows from the first, since $\rad(C)=C^\perp$ for a stabilizer code.
\end{proof}

Note that Corollary~\ref{cor:macwilliams-enum} can also be obtained by combining Theorem~\ref{thm:macwilliams-moments} with Proposition~\ref{prop:binomial-to-enumerator}. Here it follows from the duality of Proposition~\ref{prop:poly-dual}, in the same way as the MacWilliams identity for classical linear codes follows from Greene's theorem~\cite{greene1976weight}.

\begin{example}
    Let $C$ and $C^\perp$ be as in Example~\ref{example:repetition_code5}. We have $\rho_C(V^2)=\dim_{\F_2}(C)=3$, $\rho_C(A_{\{1\}})=\rho_C(A_{\{2\}})=2$, and $\rho_C(\{0\})=0$. Therefore,
    \begin{equation*}
        R_C(u,v)=u^3+2u+v\qquad\textup{ and }\qquad R_{C^\perp}(u,v)=v^3+2v+u,
    \end{equation*}
    in line with Proposition~\ref{prop:poly-dual}. By Theorem~\ref{thm:greene}, we obtain
    \begin{align*}
        \mathcal{B}(x,y;C)&=x^{1/2}(y-x)^{3/2}\left(\left(\frac{4x}{y-x}\right)^{3/2}+2\left(\frac{4x}{y-x}\right)^{1/2}+\left(\frac{y-x}{x}\right)^{1/2}\right)\\
        &=8x^2+4x(y-x)+(y-x)^2,
    \end{align*}
    which equals $y^2+2xy+5x^2$, as computed in Example~\ref{example:repetition_code5}.
\end{example}

\bibliographystyle{amsalpha}
\bibliography{references}

\newcommand{\etalchar}[1]{$^{#1}$}
\providecommand{\bysame}{\leavevmode\hbox to3em{\hrulefill}\thinspace}
\providecommand{\MR}{\relax\ifhmode\unskip\space\fi MR }
\providecommand{\MRhref}[2]{%
  \href{http://www.ams.org/mathscinet-getitem?mr=#1}{#2}
}
\providecommand{\href}[2]{#2}
\begin{thebibliography}{ACML{\etalchar{+}}24}

\bibitem[ACML{\etalchar{+}}24]{anderson2024relative}
Sarah~E. Anderson, Eduardo Camps-Moreno, Hiram~H. L{\'o}pez, Gretchen~L. Matthews, Diego Ruano, and Ivan Soprunov, \emph{Relative hulls and quantum codes}, IEEE Transactions on Information Theory \textbf{70} (2024), no.~5, 3190--3201.

\bibitem[ADH15]{almheiri2015bulk}
Ahmed Almheiri, Xi~Dong, and Daniel Harlow, \emph{Bulk locality and quantum error correction in {A}d{S}/{CFT}}, Journal of High Energy Physics \textbf{2015} (2015), no.~4, 1--34.

\bibitem[AK90]{assmus1990affine}
Edward~F. Assmus and Jennifer~D. Key, \emph{Affine and projective planes}, Discrete Mathematics \textbf{83} (1990), no.~2-3, 161--187.

\bibitem[AK94]{assmus1994designs}
\bysame, \emph{Designs and their codes}, Cambridge University Press, 1994.

\bibitem[Art16]{artin2016geometric}
Emil Artin, \emph{{G}eometric {A}lgebra}, Courier Dover Publications, 2016.

\bibitem[Bac06]{bacon2006operator}
Dave Bacon, \emph{Operator quantum error-correcting subsystems for self-correcting quantum memories}, Physical Review A—Atomic, Molecular, and Optical Physics \textbf{73} (2006), no.~1, 012340.

\bibitem[BC23]{byrne2023tensor}
Eimear Byrne and Giuseppe Cotardo, \emph{Tensor codes and their invariants}, SIAM Journal on Discrete Mathematics \textbf{37} (2023), no.~3, 1988--2015.

\bibitem[BCR20]{byrne2020rank}
Eimear Byrne, Giuseppe Cotardo, and Alberto Ravagnani, \emph{Rank-metric codes, generalized binomial moments and their zeta functions}, Linear Algebra and its Applications \textbf{604} (2020), 92--128.

\bibitem[BDH06]{brun2006correcting}
Todd Brun, Igor Devetak, and Min-Hsiu Hsieh, \emph{Correcting quantum errors with entanglement}, Science \textbf{314} (2006), no.~5798, 436--439.

\bibitem[BJMS12]{britz2012wei}
Thomas Britz, Trygve Johnsen, Dillon Mayhew, and Keisuke Shiromoto, \emph{Wei-type duality theorems for matroids}, Designs, Codes and Cryptography \textbf{62} (2012), no.~3, 331--341.

\bibitem[BSW15]{britz2015demi}
Thomas Britz, Keisuke Shiromoto, and Thomas Westerb{\"a}ck, \emph{Demi-matroids from codes over finite {F}robenius rings}, Designs, Codes and Cryptography \textbf{75} (2015), no.~1, 97--107.

\bibitem[BT09]{bravyi2009no}
Sergey Bravyi and Barbara Terhal, \emph{A no-go theorem for a two-dimensional self-correcting quantum memory based on stabilizer codes}, New Journal of Physics \textbf{11} (2009), no.~4, 043029.

\bibitem[CGLW24]{cao2024quantum}
ChunJun Cao, Michael~J Gullans, Brad Lackey, and Zitao Wang, \emph{Quantum lego expansion pack: Enumerators from tensor networks}, PRX Quantum \textbf{5} (2024), no.~3, 030313.

\bibitem[CL21]{cao2021approximate}
ChunJun Cao and Brad Lackey, \emph{Approximate bacon-shor code and holography}, Journal of High Energy Physics \textbf{2021} (2021), no.~5, 1--110.

\bibitem[CL22]{cao2022quantum}
\bysame, \emph{Quantum lego: Building quantum error correction codes from tensor networks}, PRX Quantum \textbf{3} (2022), no.~2, 020332.

\bibitem[CL23]{cao2023quantum}
\bysame, \emph{Quantum weight enumerators and tensor networks}, IEEE Transactions on Information Theory \textbf{70} (2023), no.~5, 3512--3528.

\bibitem[Cot25]{cotardo2025zeta}
Giuseppe Cotardo, \emph{Zeta functions for tensor codes}, Journal of Algebra and Its Applications \textbf{24} (2025), no.~03, 2550074.

\bibitem[CRSS98]{calderbank1998quantum}
A.~Robert Calderbank, Eric~M. Rains, Peter~M. Shor, and Neil J.~A. Sloane, \emph{Quantum error correction via codes over {GF}(4)}, IEEE Transactions on Information Theory \textbf{44} (1998), no.~4, 1369--1387.

\bibitem[DGGR25]{digiusto2025distance}
Andrea Di~Giusto, Elisa Gorla, and Alberto Ravagnani, \emph{A distance-free approach to generalized weights}, arXiv:2512.17542, 2025.

\bibitem[FCY{\etalchar{+}}04]{fattal2004entanglement}
David Fattal, Toby~S Cubitt, Yoshihisa Yamamoto, Sergey Bravyi, and Isaac~L Chuang, \emph{Entanglement in the stabilizer formalism}, arXiv preprint quant-ph/0406168 (2004).

\bibitem[FHMS21]{farrelly2021tensor}
Terry Farrelly, Robert~J Harris, Nathan~A McMahon, and Thomas~M Stace, \emph{Tensor-network codes}, Physical Review Letters \textbf{127} (2021), no.~4, 040507.

\bibitem[FMMC12]{fowler2012surface}
Austin~G. Fowler, Matteo Mariantoni, John~M. Martinis, and Andrew~N. Cleland, \emph{Surface codes: Towards practical large-scale quantum computation}, Physical Review A—Atomic, Molecular, and Optical Physics \textbf{86} (2012), no.~3, 032324.

\bibitem[FP14]{ferris2014tensor}
Andrew~J. Ferris and David Poulin, \emph{Tensor networks and quantum error correction}, Physical Review Letters \textbf{113} (2014), no.~3.

\bibitem[Got96]{gottesman1996class}
Daniel Gottesman, \emph{Class of quantum error-correcting codes saturating the quantum hamming bound}, Physical Review A \textbf{54} (1996), no.~3, 1862.

\bibitem[Got97]{gottesman1997stabilizer}
\bysame, \emph{Stabilizer codes and quantum error correction}, California Institute of Technology, 1997.

\bibitem[Gre76]{greene1976weight}
Curtis Greene, \emph{Weight enumeration and the geometry of linear codes}, Studies in Applied Mathematics \textbf{55} (1976), no.~2, 119--128.

\bibitem[GS25]{gorla2025latroids}
Elisa Gorla and Flavio Salizzoni, \emph{Latroids and code invariants}, arXiv preprint arXiv:2503.03010 (2025).

\bibitem[Har17]{harlow2017ryu}
Daniel Harlow, \emph{The {R}yu--{T}akayanagi formula from quantum error correction}, Communications in Mathematical Physics \textbf{354} (2017), no.~3, 865--912.

\bibitem[HP03]{huffman2003fundamentals}
W.~Cary Huffman and Vera Pless, \emph{Fundamentals of error-correcting codes}, Cambridge university press, 2003.

\bibitem[JP13]{jurrius2013codes}
Relinde Jurrius and Ruud Pellikaan, \emph{Codes, arrangements and matroids}, Algebraic Geometry Modeling in Information Theory, Series on Coding Theory and Cryptology, vol.~8, World Scientific, 2013, pp.~219--325.

\bibitem[Kit03]{kitaev2003fault}
Alexei~Yurievich Kitaev, \emph{Fault-tolerant quantum computation by anyons}, Annals of physics \textbf{303} (2003), no.~1, 2--30.

\bibitem[KL97]{knill1997theory}
Emanuel Knill and Raymond Laflamme, \emph{Theory of quantum error-correcting codes}, Physical Review A \textbf{55} (1997), no.~2, 900.

\bibitem[Mal99]{maldacena1999large}
Juan Maldacena, \emph{The large-{N} limit of superconformal field theories and supergravity}, International journal of theoretical physics \textbf{38} (1999), no.~4, 1113--1133.

\bibitem[NC10]{nielsen2010quantum}
Michael~A. Nielsen and Isaac~L. Chuang, \emph{Quantum computation and quantum information}, Cambridge university press, 2010.

\bibitem[Pou05]{poulin2005stabilizer}
David Poulin, \emph{Stabilizer formalism for operator quantum error correction}, Physical review letters \textbf{95} (2005), no.~23, 230504.

\bibitem[Pre98]{preskill1998quantum}
John Preskill, \emph{Quantum information and quantum computation}, Course Files on the Web: http://www. theory. caltech. edu/people/preskill/ph229 (1998).

\bibitem[PRR22]{pollack2022understanding}
Jason Pollack, Patrick Rall, and Andrea Rocchetto, \emph{Understanding holographic error correction via unique algebras and atomic examples}, Journal of High Energy Physics \textbf{2022} (2022), no.~6, 56.

\bibitem[PYHP15]{pastawski2015holographic}
Fernando Pastawski, Beni Yoshida, Daniel Harlow, and John Preskill, \emph{Holographic quantum error-correcting codes: Toy models for the bulk/boundary correspondence}, Journal of High Energy Physics \textbf{2015} (2015), no.~6, 149.

\bibitem[Rai98]{rains2002quantum}
Eric~M. Rains, \emph{Quantum weight enumerators}, IEEE Transactions on Information Theory \textbf{44} (1998), no.~4, 1388--1394.

\bibitem[Rai99]{rains1999nonbinary}
Eric~M Rains, \emph{Nonbinary quantum codes}, IEEE Transactions on Information Theory \textbf{45} (1999), no.~6, 1827--1832.

\bibitem[Rav16a]{ravagnani2016generalized}
Alberto Ravagnani, \emph{Generalized weights: an anticode approach}, Journal of Pure and Applied Algebra \textbf{220} (2016), no.~5, 1946--1962.

\bibitem[Rav16b]{ravagnani2016rank}
\bysame, \emph{Rank-metric codes and their duality theory}, Designs, codes and Cryptography \textbf{80} (2016), no.~1, 197--216.

\bibitem[SFE{\etalchar{+}}25]{steinberg2025universal}
Matthew Steinberg, Junyu Fan, Jens Eisert, Sebastian Feld, Alexander Jahn, and Chunjun Cao, \emph{Universal fault-tolerant logic with heterogeneous holographic codes}, arXiv preprint arXiv:2504.10386 (2025).

\bibitem[SL97]{shor1997quantum}
Peter Shor and Raymond Laflamme, \emph{Quantum analog of the {MacWilliams} identities for classical coding theory}, Physical review letters \textbf{78} (1997), no.~8, 1600.

\bibitem[Ver04]{vertigan2004latroids}
Dirk Vertigan, \emph{Latroids and their representation by codes over modules}, Transactions of the American Mathematical Society \textbf{356} (2004), no.~10, 3841--3868.

\bibitem[Wei91]{wei1991generalized}
Victor~K. Wei, \emph{Generalized {H}amming weights for linear codes}, IEEE Transactions on Information Theory \textbf{37} (1991), no.~5, 1412--1418.

\bibitem[Wit98]{witten1998anti}
Edward Witten, \emph{Anti de {S}itter space and holography}, Advances in Theoretical and Mathematical Physics \textbf{2} (1998), no.~2, 253--291.

\bibitem[XKH22]{xu2022galois}
Yang Xu, Haibin Kan, and Guangyue Han, \emph{A {G}alois connection approach to {W}ei-type duality theorems}, IEEE Transactions on Information Theory \textbf{68} (2022), no.~8, 5133--5144.

\end{thebibliography}
\end{document}